\documentclass[final]{siamart220329}

\crefname{section}{Section}{Sections} 
\crefname{subsection}{Subsection}{Subsections} 

\pdfpagesattr{/CropBox [0 60 515 740]}

\usepackage{amsfonts}
\usepackage{graphicx}

\usepackage{mathtools}
\usepackage[nocompress,noadjust]{cite}

\usepackage{tikz}

\definecolor{darkred}{rgb}{0.75,0,0}

\definecolor{slgray}{rgb}{0.92,0.92,0.92} 

\newcommand{\signature}{\sigma}
\newcommand{\shape}{\bar{\sigma}}

\newcommand{\cw}{\mathsf{cw}} 
\newcommand{\tw}{\mathsf{tw}} 
\newcommand{\td}{\mathsf{td}} 
\newcommand{\vc}{\mathsf{vc}} 
\newcommand{\nd}{\mathsf{nd}} 
\newcommand{\cd}{\mathsf{cd}} 
\newcommand{\tc}{\mathsf{tc}} 

\newcommand{\dist}{\operatorname{dist}}

\newcommand{\bigO}{\mathcal{O}}

\makeatletter
\newcommand*{\bottop}{%
  \mathpalette\@bottop{}%
}
\newcommand*{\@bottop}[2]{%
  \rlap{$#1\bot\m@th$}
  \top
}
\makeatother

\newcommand{\mso}[1]{$\mathsf{MSO}_{#1}$}

\newcommand{\MSORfull}[1]{\mso{#1}-\textsc{Reconfiguration}}

\newcommand{\MSOR}[1]{\mso{#1}-$\mathrm{R}$}

\newcommand{\TJ}{\mathsf{TJ}}

\newcommand{\numq}{\mathsf{q}}

\newenvironment{listing}[1]{%
  \begin{list}{*}{%
      \settowidth{\labelwidth}{#1}%
      \setlength{\leftmargin}{\labelwidth}%
      \advance \leftmargin by 12pt
      \setlength{\itemsep}{0pt}%
      \setlength{\parsep}{0pt}%
      \setlength{\topsep}{0pt}%
      \setlength{\parskip}{0pt}%
    }%
  }{%
  \end{list}
}

\makeatletter
\usepackage{tabularx,environ}
\makeatletter
\newcommand{\problemtitle}[1]{\gdef\@problemtitle{#1}}
\newcommand{\probleminput}[1]{\gdef\@probleminput{#1}}
\newcommand{\problemquestion}[1]{\gdef\@problemquestion{#1}}
\NewEnviron{myproblem}{
  \problemtitle{}\probleminput{}\problemquestion{}
  \BODY
  \par\addvspace{.5\baselineskip}
  \centering
  \fbox{
    \parbox{0.9\hsize}{
      \@problemtitle
      \begin{listing}{\textbf{Question:}}
	\item[\textbf{Input:}] \@probleminput
	\item[\textbf{Question:}] \@problemquestion
      \end{listing}
    }
  }
  \par\addvspace{.5\baselineskip}
}
\makeatother

\newsiamremark{claim}{Claim}
\newsiamremark{observation}{Observation}

\headers{Meta-Theorems for Combinatorial Reconfiguration}
{T. Gima, T. Ito, Y. Kobayashi, and Y. Otachi}

\title{Algorithmic Meta-Theorems for Combinatorial Reconfiguration Revisited\thanks{%
Submitted to the editors \today.
A preliminary version appeared in the proceedings of 
the 30th European Symposium on Algorithms (ESA 2022),
Leibniz International Proceedings in Informatics 244 (2022) 61:1--61:15.
\funding{This work was partially supported by JSPS KAKENHI Grant Numbers
  JP18H04091, 
  JP19K11814, 
  JP20H00595, 
  JP20H05793, 
  JP20K19742,
  JP21K11752, 
  JP22H00513. 
}}}

\author{%
  Tatsuya Gima\thanks{%
    Nagoya University, Nagoya, Japan 
    (\email{gima@nagoya-u.jp}, \email{otachi@nagoya-u.jp}).
  }
  \and
  Takehiro Ito\thanks{%
    Graduate School of Information Sciences, Tohoku University, Sendai, Japan
    (\email{takehiro@tohoku.ac.jp}).
  }
  \and
  Yasuaki Kobayashi\thanks{%
    Hokkaido University, Sapporo, Japan
    (\email{koba@ist.hokudai.ac.jp}).
  }
  \and
  Yota Otachi\footnotemark[2]
}

\ifpdf
\hypersetup{
  pdftitle={Algorithmic Meta-Theorems for Combinatorial Reconfiguration Revisited},
  pdfauthor={T. Gima, T. Ito, Y. Kobayashi, and Y. Otachi}
}
\fi

\begin{document}

\maketitle

\begin{abstract}
Given a graph and two vertex sets satisfying a certain feasibility condition,
a reconfiguration problem asks whether we can reach one vertex set from the other
by repeating prescribed modification steps while maintaining feasibility.
In this setting, Mouawad et al.~[IPEC 2014] presented
an algorithmic meta-theorem for reconfiguration problems that says
if the feasibility can be expressed in monadic second-order logic (MSO),
then the problem is fixed-parameter tractable parameterized by $\textrm{treewidth} + \ell$,
where $\ell$ is the number of steps allowed to reach the target set.
On the other hand, it is shown by Wrochna~[J.\ Comput.\ Syst.\ Sci.\ 2018]
that if $\ell$ is not part of the parameter,
then the problem is PSPACE-complete even on graphs of constant bandwidth.

In this paper, we present the first algorithmic meta-theorems for the case where $\ell$ is not part of the parameter,
using some structural graph parameters incomparable with bandwidth.
We show that if the feasibility is defined in MSO,
then the reconfiguration problem under the so-called token jumping rule is
fixed-parameter tractable parameterized by neighborhood diversity.
We also show that the problem is fixed-parameter tractable parameterized by $\textrm{treedepth} + k$,
where $k$ is the size of sets being transformed.
We finally complement the positive result for treedepth by showing that
the problem is PSPACE-complete on forests of depth~$3$.
\end{abstract}




\section{Introduction}
\label{sec:intro}

A reconfiguration problem asks, given two feasible solutions $S$ and $S'$ of a combinatorial problem,
whether there is a step-by-step transformation from $S$ to $S'$ without losing the feasibility~\cite{ItoDHPSUU11}.
The field studying such problems, called \emph{combinatorial reconfiguration}, is growing rapidly.
The source combinatorial problems in reconfiguration problems have spread in many subareas of theoretical computer science
(see surveys~\cite{Heuvel13,Nishimura18}).
In this work, we focus on reconfiguration problems on graphs,
especially the ones considering some vertex subsets as feasible solutions.
Such problems involve classic properties like
independent sets~\cite{KaminskiMM12},
vertex covers~\cite{MouawadNRS18},
dominating sets~\cite{SuzukiMN16},
and some connected variants~\cite{LokshtanovMPS22}.
Restrictions to some important graph classes 
such as bipartite graphs~\cite{LokshtanovM19}, split graphs~\cite{BelmonteKLMOS21}, and sparse graphs~\cite{LokshtanovMPRS18}
are also studied.

Since many problems are studied under many settings in combinatorial reconfiguration,
one may ask for a unified method, or an \emph{algorithmic meta-theorem}, for handling reconfiguration problems like Courcelle's theorem 
for classic (non-reconfiguration) problems~\cite{Courcelle90mso1,Courcelle92mso3,ArnborgLS91,BoriePT92,CourcelleMR00}.
Since reconfiguration problems are hard in general (often PSPACE-complete~\cite{ItoDHPSUU11}),
we need to consider some special cases or introduce some additional parameters to consider fixed-parameter tractability.
One successful approach in this direction was taken by Mouawad et al.~\cite{MouawadNRW14}, who showed that
if the feasible solutions in a graph can be expressed in monadic second-order logic,
then the reconfiguration problem (under reasonable transformation rules)
is fixed-parameter tractable parameterized simultaneously by the treewidth of the underlying graph 
and the length of a transformation sequence. 
Their method is quite general and can be applied to several other settings.\footnote{%
We elaborate on this a little more in \cref{ssec:related}.}
On the other hand, Wrochna~\cite{Wrochna18} showed that if the length of a transformation sequence is not part of the parameter,
then some problems that fit in this framework are PSPACE-complete even on graphs of constant bandwidth.

The two results mentioned above
(the tractability parameterized by treewidth $+$ transformation length~\cite{MouawadNRW14} and 
the intractability parameterized solely by bandwidth~\cite{Wrochna18}) might be interpreted as that
if we have the length of a transformation sequence in the parameter, then we can do pretty much everything we expect,
and otherwise we can expect very little.
Thus, one might conclude that this line of research is complete and
the length of a transformation sequence is necessary and sufficient in some sense for having efficient algorithms.
Indeed, to the best of our knowledge, the study of algorithmic meta-theorems for reconfiguration problems was not extended after these results.

In this paper, we revisit the investigation of algorithmic meta-theorems for reconfiguration problems
and shed light on the settings where the length of a transformation sequence is \emph{not} part of the parameter.
In particular, we present fixed-parameter algorithms
for the reconfiguration problem of vertex sets defined by a monadic second-order formula
parameterized by vertex cover number or neighborhood diversity.
We also show that when combined with the solution set size,
treedepth can be used to obtain a fixed-parameter algorithm.
We then complement this result by showing that when the solution size is not part of the parameter,
the problem is PSPACE-complete on graphs of constant treedepth.

\subsection{Our results}
Now we give a little more precise description of our results.
Formal definitions not given here can be found in \cref{sec:pre}.

For a graph $G$, we denote its clique-width by $\cw(G)$, 
treewidth by $\tw(G)$,
treedepth by $\td(G)$, vertex cover number by $\vc(G)$,
neighborhood diversity by $\nd(G)$,
cluster deletion number by $\cd(G)$.
(We define some of these parameters in the last part of \cref{sec:pre}.)
See \cref{fig:parameters} for the hierarchy among the graph parameters studied in this paper and some related ones.
For a graph parameter $\mathsf{f}$, we often say informally that
a problem is fixed-parameter tractable ``parameterized by $\mathsf{f}$''
to mean ``parameterized by $\mathsf{f}(G)$, where $G$ is the input graph.''

Given a monadic second-order (\mso{1}) formula  $\phi$ with one free set variable,
a graph $G$, and two vertex subsets $S, S'$ of the same size,
\MSORfull{1} (\MSOR{1}) asks whether there exists a sequence of vertex subsets from $S$ to $S'$
such that each set in the sequence satisfies the property expressed by $\phi$
and each set in the sequence is obtained from the previous one by exchanging one vertex with another.
Note that this rule allows to exchange any pair of vertices.
Such a rule is well studied and called the \emph{token jumping} rule~\cite{KaminskiMM12}.
There is another well-studied rule called the \emph{token sliding} rule~\cite{HearnD05},
which requires that the exchanged vertices are adjacent in $G$.
We focus on the simpler rule token jumping in this paper
and comment on the token sliding counter parts in \cref{sec:token-sliding}.
\MSORfull{2} (\MSOR{2}) with more general \mso{2} formulas is defined analogously.

To show a concrete example of \MSOR{1}, 
let  $\phi(S) \coloneqq \forall u \, \forall v \colon (u \in S \land v \in S) \Rightarrow \lnot E(u,v)$.
This $\phi$ is an \mso{1} formula (see \cref{sec:pre}) expressing that $S$ is an independent set.
Thus, \MSOR{1} with this $\phi$ is exactly \textsc{Independent Set Reconfiguration} under the token jumping rule.

Now the main results in this paper can be summarized as follows.
\begin{enumerate}
  \item \MSOR{1} is fixed-parameter tractable parameterized by $\nd + |\phi|$.
  \begin{itemize}
    \item \MSOR{2} is fixed-parameter tractable parameterized by $\vc + |\phi|$,
    but not by $\nd + |\phi|$ unless $\mathrm{E} = \mathrm{NE}$.
    \item The positive results here strongly depend on the token jumping rule.
  \end{itemize}

  \item \MSOR{2} is fixed-parameter tractable parameterized by $\td + k + |\phi|$, where $k$ is the size of input sets $S$ and $S'$.
  \begin{itemize}
    \item This result holds also under the token sliding rule.
    \item As a by-product, we show that \MSOR{1} is fixed-parameter tractable parameterized by $\cd + k + |\phi|$.
  \end{itemize}

  \item For some fixed $\phi$, \MSOR{1} is PSPACE-complete even on forests of depth~$3$.
  \begin{itemize}
    \item A similar hardness result can be shown under the token sliding rule.
  \end{itemize}
\end{enumerate}
In all positive results, we can find a shortest sequence for transformation if any exists.

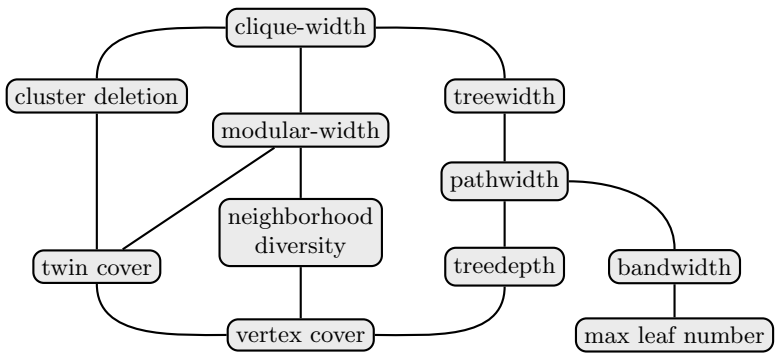
\begin{figure}[tb]
  \centering
\begin{tikzpicture}[every node/.style={draw, thick, rectangle, rounded corners,align=center,fill=slgray},scale=.90]
  \small

  \node (vc) at (0,   0) {vertex cover};
  \node (ml) at (5.5, 0) {max leaf number};

  \node (nd) at (0,   1.5) {neighborhood \\ diversity};
  \node (td) at (3,   1) {treedepth};
  \node (bw) at (5.5, 1) {bandwidth};
  \node (tc) at (-3,  1) {twin cover};

  \node (mw) at (0,  3) {modular-width};
  \node (pw) at (3,  2.25) {pathwidth};

  \node (cd) at (-3, 3.5) {cluster deletion};
  \node (tw) at (3,  3.5) {treewidth};

  \node (cw) at (0, 4.5) {clique-width};

  \draw[thick] (cw) -- (mw) -- (nd) -- (vc);
  \draw[thick] (tw) -- (pw) -- (td);

  \draw[thick] (cw) [out=0,in=90] to (tw);
  \draw[thick] (cw) [out=180,in=90] to (cd);

  \draw[thick] (td) [out=270,in=0] to (vc);
  \draw[thick] (tc) [out=270,in=180] to (vc);
  
  \draw[thick] (cd) -- (tc);

  \draw[thick] (tc) -- (mw);

  \draw[thick] (bw) -- (ml);
  \draw[thick] (pw) [out=0,in=90] to (bw);

\end{tikzpicture}
  \caption{The graph parameters studied in this paper (with some related ones).
    A connection between two parameters indicates the existence of a function in the one above 
    that lower-bounds the one below.}
  \label{fig:parameters}
\end{figure}

\subsection{Related work}
\label{ssec:related}
Wrochna~\cite{Wrochna18} showed that \MSOR{1} is PSPACE-complete
on graphs of constant bandwidth when $\phi$ expresses independent sets.
This implies the PSPACE-completeness of \MSOR{1} on graphs of
constant pathwidth, treewidth, and clique-width (see \cref{fig:parameters}).
This result was later strengthened by van der Zanden~\cite{Zanden15},
who showed that the same problem is PSPACE-complete on planar graphs having constant bandwidth and maximum degree~$3$.

To cope with this intractability,
Mouawad et al.~\cite{MouawadNRW14} considered a variant with the additional restriction that 
the length of a transformation sequence cannot exceed some upper bound $\ell$.
They showed that this variant of \MSOR{2} is fixed-parameter tractable parameterized by $\ell + \tw + |\phi|$.
They reduce the reconfiguration problem 
to the model-checking problem of a single \mso{2} formula 
by expressing the existence of fewer than $\ell$ intermediate sets that satisfy $\phi$
and also expressing that the change from a set to the next one obeys the transformation rule.
Their framework is quite general and can be used in several other settings
such as vertex sets defined by an \mso{1} formula with the parameter $\ell + \cw + |\phi|$,
or size-$k$ vertex sets defined by a first-order formula with the parameter $\ell + k + |\phi|$ on a nowhere dense graph class (as observed also in \cite{LokshtanovMPRS18}).
Also, since the step-by-step modification can be defined by a formula, 
the results apply not only for the token jumping rule but also for several other rules including the token sliding rule.\footnote{%
Actually, the rule used in~\cite{MouawadNRW14} was another one called ``token addition and removal,'' which is trickier to handle.}

Another important line of studies on parameterized complexity of reconfiguration problems
take the input set size $k$ as the main parameter instead of a graph structural parameter.
This line was initiated by Mouawad et al.~\cite{MouawadN0SS17},
who showed several results parameterized solely by $k$ and also by $k + \ell$.
See the recent survey by Bousquet et al.~\cite{BousquetMNS22arxiv}.
Recently, Bodlaender et al.~\cite{BodlaenderGNS21,BodlaenderGS21} 
further extended this line by showing that
depending on whether and how the parameter depends on $\ell$,
the problem becomes complete to XL, XNL, or XNLP\@.


\section{Preliminaries}
\label{sec:pre}
We assume that the reader is familiar with the parameterized complexity theory.
See a standard textbook (e.g., \cite{CyganFKLMPPS15,Niedermeier06,FlumG06,DowneyF13}) for basic definitions.

Let $G = (V,E)$ be a graph.
For $X \subseteq V$, we denote by $G[X]$ and $G-X$ the graphs induced by $X$ and $V \setminus X$, respectively.
We sometimes denote the vertex set of $G$ by $V(G)$ and the edge set by $E(G)$.
For a digraph $D$, we denote by $A(D)$ its arc set.

For a non-negative integer $d$, let $[d]$ denote the set $\{i \in \mathbb{Z} \mid 1 \le i \le d\}$.
For two non-negative integers $a,b$ with $a \le b$, let $[a,b]$ denote the set $\{i \in \mathbb{Z} \mid a \le i \le b\}$.

\subparagraph*{Colored graphs}

In this paper, we consider graphs in which each vertex has a (possibly empty) set of colors.
We call them \emph{colored graphs}.
Formally, a colored graph $G$ is a tuple $(V,E,\mathcal{C})$
such that the vertex set is $V$,
the edge set $E \subseteq \binom{V}{2}$ is a set of unordered pairs of vertices,
and
$\mathcal{C} = \langle C_{1}, \dots, C_{c} \rangle$
is a tuple of subsets of $V$, where each $C_{i}$ is called a \emph{color}.
For $v \in V$, let $\mathcal{C}(v)$ denote the set of colors that $v$ belongs to.
When $\mathcal{C}(v) = \emptyset$ for all $v \in V$, then the graph is \emph{uncolored}.
As we describe later, monadic second-order formulas treat
the edge set as a symmetric binary relation on $V$ and each color as a unary relation on $V$.
As the number of colors a formula $\phi$ can access is bounded by $|\phi|$,
which is always considered as a parameter or a constant in this paper,
we can assume that the number of colors $c$ is a parameter as well.
We omit the information of colors and say $G = (V,E)$ when colors do not matter.

Two colored graphs $G = (V, E, \langle C_{1}, \dots, C_{c} \rangle)$ and $G' = (V', E', \langle C'_{1}, \dots, C'_{c} \rangle)$ are \emph{isomorphic}
if there is a color-preserving isomorphism $f \colon V \to V'$; that is,
\begin{itemize}
  \item for all $u,v \in V$, $\{u,v\} \in E$ if and only if $\{f(u),f(v)\} \in E'$, and
  \item for all $v \in V$ and $1 \le i \le c$, $v \in C_{i}$ if and only if $f(v) \in C'_{i}$.
\end{itemize}
We also say that 
$\langle G, S \rangle$ and $\langle G', S' \rangle$ are isomorphic for sets $S \subseteq V$ and $S' \subseteq V'$ if the colored graphs 
$(V, E, \langle C_{1}, \dots, C_{c}, S \rangle)$ and $(V', E', \langle C'_{1}, \dots, C'_{c}, S' \rangle)$ are isomorphic.

\subparagraph*{Monadic second-order logic}
In the monadic second-order logic on colored graphs, denoted \mso{1},
we can use vertex variables and vertex-set variables.
The atomic formulas are 
the equality $x=y$ of vertex variables,
the adjacency relation $E(x,y)$ which means $\{x,y\} \in E$,
the color predicate $C_{i}(x)$ for each color $C_{i}$ which means $x \in C_{i}$,
and the inclusion predicate $X(x)$ for a variable $x$ and a set variable $X$ which means $x \in X$.
The \mso{1} formulas are recursively defined from atomic formulas
using the usual Boolean connectives ($\lnot$, $\land$, $\lor$, $\Rightarrow$, $\Leftrightarrow$),
and quantification of variables ($\forall x$, $\exists x$, $\forall X$, $\exists X$).
For the sake of readability, we often use syntactic sugar in \mso{1} formulas 
(e.g., we write ``$\exists x \in X\colon \psi$'' to mean ``$\exists x \colon X(x) \land \psi$'').
As syntax sugars, we also use dotted quantifiers $\dot{\exists}$ and $\dot{\forall}$ to quantify distinct objects.
For example, $\dot{\exists} a,b,c\colon \psi$
means $\exists a,b,c\colon (a \ne b) \land (b \ne c) \land (c \ne a) \land \psi$ and
$\dot{\forall} a,b,c\colon \psi$
means $\forall a,b,c\colon ((a \ne b) \land (b \ne c) \land (c \ne a)) \Rightarrow \psi$.

\mso{2} is an extension of \mso{1} that additionally allows edge variables, edge-set variables,
and an atomic formula $I(e,x)$ that represents the edge-vertex incidence relation.
It is known that \mso{2} is strictly more powerful than \mso{1} in general~\cite{CourcelleE12}.

An \mso{1} (or \mso{2}) formula $\phi$ with free variables $X_{1}, \dots, X_{p}$
is denoted by $\phi(X_{1}, \dots, X_{p})$.
For a graph $G$ and vertex subsets $S_{1}, \dots, S_{p}$ of $G$,
we write $G \models \phi(S_{1}, \dots, S_{p})$
if $\phi$ is true for $G$ when the free variables $X_{1}, \dots, X_{p}$ are interpreted as $S_{1}, \dots, S_{p}$.
We call an \mso{1} (\mso{2}) formula without free variables an \mso{1} (\mso{2}, resp.) \emph{sentence}.
\begin{proposition}
[Folklore, see e.g., \cite{LampisM21}] 
\label{prop:isomorphic}
Let $G$ and $G'$ be colored graphs and $S$ and $S'$ be some vertex subsets of them
such that $\langle G, S \rangle$ and $\langle G', S' \rangle$ are isomorphic.
Then, for every \mso{1} (or \mso{2}) formula $\phi$ with one free set variable, 
$G \models \phi(S)$ if and only if $G' \models \phi(S')$.
\end{proposition}

\subparagraph*{Problem definitions}
For a colored graph $G$ and an \mso{1} (or \mso{2}) formula $\phi(X)$
a sequence $S_{0}, \dots, S_{\ell}$ of vertex subsets of $G$ is a \emph{$\TJ(\phi)$-sequence from $S_{0}$ to $S_{\ell}$} (of length $\ell$) if 
\begin{itemize}
  \item $|S_{i-1} \setminus S_{i}| = |S_{i} \setminus S_{i-1}| = 1$ for every $i \in [\ell]$, and
  \item $G \models \phi(S_{i})$ for $0 \le i \le \ell$.
\end{itemize}
We denote by $\dist_{\phi,G}(S,S')$ the minimum length of a $\TJ(\phi)$-sequence from $S$ to $S'$,
which is set to $\infty$ if there is no such sequence.
We call a $\TJ(\phi)$-sequence of length~$1$ a \emph{$\TJ(\phi)$-move}.
Now the main problem studied in this paper can be formalized as follows.
\begin{myproblem}
  \problemtitle{\MSORfull{1} (\MSOR{1})}
  \probleminput{An \mso{1} formula $\phi$, a colored graph $G = (V,E,\mathcal{C})$, and sets $S, S' \subseteq V$ such that
  $|S| = |S'|$, $G \models \phi(S)$, and $G \models \phi(S')$.}
  \problemquestion{Is there a $\TJ(\phi)$-sequence from $S$ to $S'$?}
\end{myproblem}
We also study \MSOR{2} that allows \mso{2} formulas having one free vertex-set variable as $\phi$.
Observe that \MSOR{1} is PSPACE-hard as it generalizes various PSPACE-complete reconfiguration problems such as 
\textsc{Independent Set Reconfiguration}.
On the other hand, it still belongs to PSPACE since we can non-deterministically find the next vertex set $R$ in the $\TJ(\phi)$-sequence
and test whether $G \models \phi(R)$ holds in PSPACE~\cite{Stockmeyer71phdthesis,Vardi82}.

When describing a $\TJ(\phi)$-move from $S_{i-1}$ to $S_{i}$, it is sometimes convenient to say that a \emph{token} on
the vertex $u \in S_{i-1} \setminus S_{i}$ is moved to the vertex $v \in S_{i} \setminus S_{i-1}$.
The intuition behind this is that a vertex set in a $\TJ(\phi)$-sequence is considered as the positions of tokens
and that in one $\TJ(\phi)$-move, one token on some vertex jumps to another vertex.
For simplicity, we write $S_{i-1} -u +v$ instead of $(S_{i-1} \setminus \{u\}) \cup \{v\}$.

\subparagraph*{Graph parameters}
For a graph $G = (V,E)$, a set $S \subseteq V$ is a \emph{vertex cover}
if each $e \in E$ has at least one endpoint in $S$.
The \emph{vertex cover number} of $G$, denoted $\vc(G)$, 
is the size of a minimum vertex cover of $G$.
A vertex cover of size $k$ of an $n$-vertex graph, if any exists,
can be found in time $\bigO(c^{k} \cdot n)$ for some small constant $c$ (see e.g.,~\cite{ChenKX10}). 
This implies that we can assume that a vertex cover of minimum size is given with the input
when $\vc(G)$ is part of the parameter.

Two vertices $u$ and $v$ are \emph{twins} if $N(u) = N(v)$ or $N[u] = N[v]$.
The \emph{neighborhood diversity} of a graph $G$, denoted $\nd(G)$, 
is the number of subsets $V_{i}$ in the unique partition $V_{1}, \dots, V_{p}$ of $V$
into maximal sets of twin vertices.
It is known that the neighborhood diversity and 
the corresponding partition can be computed in linear time~\cite{McConnellS99,TedderCHP08}.
From the definitions, we can see that $\nd(G) \le 2^{\vc(G)} + \vc(G)$ for every graph $G$~\cite{Lampis12}.

The \emph{treedepth} of a graph $G = (V,E)$, denoted $\td(G)$, is the minimum depth $d$ of 
a rooted forest $F$ on the vertex set $V$ such that each edge of $G$ connects an ancestor and a descendant in $F$,
where the depth of a forest is defined as the maximum distance between a root and a leaf plus~$1$.
We call such a forest a \emph{treedepth decomposition}.
It is known that a treedepth decomposition of depth $d$, if exists, can be found in time $2^{\bigO(d^{2})} \cdot n$~\cite{ReidlRVS14}.
Thus we may assume that a treedepth decomposition of depth $\td(G)$ is given with the input
when $\td(G)$ is part of the parameter.


\section{\MSOR{1} parameterized by neighborhood diversity}
\label{sec:vc-nd}

The main result of this section is the following theorem.
\begin{theorem}
\label{thm:msor1-nd}
\MSOR{1} parameterized by $\nd + |\phi|$ is fixed-parameter tractable.
Furthermore, for a yes instance of \MSOR{1}, finding a shortest $\TJ(\phi)$-sequence is fixed-parameter tractable with the same parameter.
\end{theorem}

We prove \cref{thm:msor1-nd} in \cref{sec:shortestND}
and then discuss the possibility of an extension to \MSOR{2} in \cref{mso2ND}.
To prove \cref{thm:msor1-nd},
we first partition the feasible sets into a small number of equivalence classes.
We show that the reachability between feasible sets can be 
checked by using an appropriately defined adjacency between the equivalence classes.
Then, we take a deeper look at the connections between the equivalence classes
and show that a shortest reconfiguration sequence can be found by finding 
some flow-like structure among the equivalence classes.

\subsection{Finding a shortest sequence}
\label{sec:shortestND}
In this section, we fix the input of \MSOR{1} as follows:
\begin{itemize}
  \item $\phi(X)$: an \mso{1} formula with one free set variable $X$;
  \item $G = (V,E,\mathcal{C})$: a colored graph;
  \item $S, S' \subseteq V$: the initial and target sets such that $G \models \phi(S)$, $G \models \phi(S')$, and $|S| = |S'| = k$.
\end{itemize}
We say that a set $X \subseteq V$ is \emph{feasible} if $G \models \phi(X)$.

We assume that the sets $S$ and $S'$ are colors in $G$;
that is, $\mathcal{C}$ is of the form like $\langle C_{1}, \dots, C_{c}, S, S' \rangle$.
If $S$ and $S'$ are originally not colors in $G$, we may add them and increase the number of colors only by~$2$.

Two vertices $u,v \in V$ in $G$ are of the same \emph{type}
if $u$ and $v$ are twins and $\mathcal{C}(u) = \mathcal{C}(v)$.
Let $\langle V_{1}, \dots, V_{t} \rangle$ be the partition of $V$ into the sets of vertices of the same type.
We call each $V_{i}$ a \emph{type}. 
Note that the type partition can be computed in polynomial time
and that $t$ depends only on the neighborhood diversity of $G$ and the number of colors in $\mathcal{C}$.

For an \mso{1} formula $\psi$, let $\numq(\psi) = 2^{\numq_\text{s}} \cdot \numq_{\text{v}}$,
where $\numq_{\text{s}}$ and $\numq_{\text{v}}$ are the numbers of set and vertex quantifies in $\psi$, respectively.
Lampis~\cite{Lampis12} proved the following fact, which is one of the main ingredients in our algorithm.\footnote{%
Note that \cref{prop:Lampis12} implies that,
when neighborhood diversity is part of the parameter,
the \mso{1} model-checking problem admits a small induced subgraph of the input graph
as a kernel~\cite{Lampis12}.
However, this does not directly show the fixed-parameter tractability of \MSOR{1} 
(let alone the stronger claim of \cref{thm:msor1-nd}).
In fact, as we will see later, an analogous result (\cref{prop:many-sets}) that implies
a ``natural'' kernel for the \mso{1} model-checking problem parameterized by treedepth is known,
while \MSOR{1} is PSPACE-complete on graphs of constant treedepth (\cref{thm:pspcae-c_td=3}).}
\begin{proposition}
[\cite{Lampis12}]
\label{prop:Lampis12}
Let $\psi$ be an \mso{1} sentence.
Assume that a graph $H$ has more than $\numq(\psi)$ vertices of the same type, and
$H'$ is the graph obtained from $H$ by removing a vertex in that type.
Then, $H \models \psi$ if and only if $H' \models \psi$.
\end{proposition}

We need the concept of ``shapes'' of vertex subsets that was used with \cref{prop:Lampis12}
in the context of extended \mso{1} model-checking problems~\cite{KnopKMT19}.
Here we introduce it in the following simplified form, which is sufficient for our purpose.
The \emph{signature} of $X \subseteq V$ is the mapping $\signature_{X} \colon [t] \to \mathbb{Z}_{\ge 0}$
such that $\signature_{X}(i) = |V_{i} \cap X|$.
A \emph{shape} is a mapping from $[t]$ to $\mathbb{Z}_{\ge 0} \cup \{\bottop\}$
that maps each $i \in [t]$ to an element of $[0,\numq(\phi)-1] \cup \{\bottop\} \cup [|V_{i}| - \numq(\phi) + 1, |V_{i}|]$.
Note that the number of shapes is $(2\numq(\phi)+1)^{t}$.
A set $X \subseteq V$ \emph{has shape} $\shape$ if for every $i \in [t]$,
\[
  \shape(i) = 
  \begin{cases}
    \bottop  & \numq(\phi) \le \signature_{X}(i) \le |V_{i}| - \numq(\phi), \\
    \signature_{X}(i) & \text{otherwise}.
  \end{cases}
\]
We say that a shape $\shape$ is \emph{$k$-feasible}
if there is a feasible set $X \subseteq V$ of size $k$ that has $\shape$ as its shape.

Let $\shape_{S}$ and $\shape_{S'}$ be the shapes of the input sets $S$ and $S'$, respectively.
Since $S$ is a color of $G$, each type $V_{i}$ either is a subset of $S$ or has no intersection with $S$,
that is,
\[
  \shape_{S}(i) = 
  \begin{cases}
    |V_{i}| & V_{i} \subseteq S, \\
    0       & \text{otherwise}.
  \end{cases}
\]
This implies that a set $R \subseteq V$ has shape $\shape_{S}$ if and only if $R = S$.
This applies to $S'$ as well.
\begin{observation}
\label{obs:nd-rshape}
$R \subseteq V$ has shape $\shape_{S}$ ($\shape_{S'}$)
if and only if $R = S$ ($R = S'$, resp.).
\end{observation}

\cref{prop:Lampis12} and the definition of shapes together give the following fact,
which is known in more general forms in the previous studies (see e.g., \cite{KnopKMT19}).
This less general one is sufficient in our setting.
We present a full proof here to be self contained.
\begin{lemma}
\label{lem:nd-shape}
If $R, R' \subseteq V$ have the same shape,
then $G \models \phi(R)$ if and only if $G \models \phi(R')$.
\end{lemma}
\begin{proof}
Let $\phi_{X}$ be the \mso{1} sentence obtained from $\phi(X)$ by considering the free variable $X$ in $\phi(X)$ as a new color.
Note that $\numq(\phi_{X})= \numq(\phi)$.
Let $G_{R}$ and $G_{R'}$ be the graphs obtained from $G$ by considering $R$ and $R'$ as realizations of the new color $X$, respectively.
Now it suffices to show that $G_{R} \models \phi_{X}$ if and only if $G_{R'} \models \phi_{X}$.

Recall that $\langle V_{1}, \dots, V_{t} \rangle$ is the type partition of $G$.
Observe that as their type partitions,
$G_{R}$ has $\langle V_{1} \cap R, \dots, V_{t} \cap R, V_{1} \setminus R, \dots, V_{t} \setminus R \rangle$ and
$G_{R'}$ has $\langle V_{1} \cap R', \dots, V_{t} \cap R', V_{1} \setminus R', \dots, V_{t} \setminus R' \rangle$,
where some of the types may be empty.
Since $R$ and $R'$ have the same shape, it holds for each $i$ that
\begin{itemize}
  \item $|V_{i} \cap R| \ge \numq(\phi)$ if and only if $|V_{i} \cap R'| \ge \numq(\phi)$, and
  \item $|V_{i} \setminus R| \ge \numq(\phi)$ if and only if $|V_{i} \setminus R'| \ge \numq(\phi)$.
\end{itemize}

We remove vertices from $G_{R}$ and $G_{R'}$
by applying \cref{prop:Lampis12} repeatedly as long as some type has more than $\numq(\phi)$ vertices.
We call the resultant graphs $G_{R}^{*}$ and $G_{R'}^{*}$.
Observe that each type in $G_{R}^{*}$ is a subset of one in $G_{R}$.
If a type in $G_{R}$ has at most $\numq(\phi)$ vertices, then the type stays the same in $G_{R}^{*}$.
On the other hand, if a type in $G_{R}$ has more than $\numq(\phi)$ vertices,
then the type becomes smaller and has size exactly $\numq(\phi)$ in $G_{R}^{*}$.
The same holds for $G_{R'}$ and $G_{R'}^{*}$ as well.
Now the discussion in the previous paragraph implies that for each $i$,
\begin{itemize}
  \item $|(V_{i} \cap R) \cap V(G_{R}^{*})| = |(V_{i} \cap R') \cap V(G_{R'}^{*})|$ and
  \item $|(V_{i} \setminus R) \cap V(G_{R}^{*})| = |(V_{i} \setminus R') \cap V(G_{R'}^{*})|$.
\end{itemize}
This implies that $G_{R}^{*}$ and $G_{R'}^{*}$ are isomorphic.
Thus the lemma holds by \cref{prop:isomorphic}.
\end{proof}

\cref{lem:nd-shape} implies in particular that if a shape $\shape$ is $k$-feasible,
then every size-$k$ set of shape $\shape$ is feasible.
\begin{lemma}
\label{lem:nd-sameshape}
If feasible sets $R, R' \subseteq V$ have the same shape and size,
then there is a $\TJ(\phi)$-sequence of length $|R \setminus R'|$ from $R$ to $R'$ 
such that all sets in the sequence have the same shape.
\end{lemma}
\begin{proof}
We use induction on $|R \setminus R'|$.
If $|R \setminus R'| = 0$, then we are done.
Assume that $|R \setminus R'| \ge 1$
and that the statement of the lemma is true for sets with strictly smaller differences.

First assume that $R$ and $R'$ have the same signature.
Then, for some $i$, it holds that $|V_{i} \cap R| = |V_{i} \cap R'|$ and $V_{i} \cap R \ne V_{i} \cap R'$.
Let $u \in V_{i} \cap (R \setminus R')$ and $v \in V_{i} \cap (R' \setminus R)$.
Then $R - u + v$ has the same signature (and thus the same shape) as $R$ and $R'$.
By \cref{lem:nd-shape}, $R - u + v$ is feasible.
By the induction hypothesis, the lemma holds.

Next assume that $R$ and $R'$ have different signatures.
Let $\shape$ be the (common) shape of $R$ and $R'$.
Since $|R| = |R'|$, there exist two indices $i$ and $j$ such that 
$|V_{i} \cap R| > |V_{i} \cap R'|$ and $|V_{j} \cap R| < |V_{j} \cap R'|$.
This implies that $\shape(i) = \shape(j) = \bottop$, and thus
\begin{itemize}
  \item $|V_{i} \cap R'|$, $|V_{j} \cap R|$, $|V_{i} \setminus R|$, $|V_{j} \setminus R'| \ge \numq(\phi)$, and
  \item $|V_{i} \cap R|$, $|V_{j} \cap R'|$, $|V_{i} \setminus R'|$, $|V_{j} \setminus R| \ge \numq(\phi) + 1$.
\end{itemize}
Now let $u \in V_{i} \cap (R \setminus R')$ and $v \in V_{j} \cap (R' \setminus R)$.
Then $R - u + v$ has the shape $\shape$,
and thus it is feasible by \cref{lem:nd-shape}.
By the induction hypothesis, the lemma holds.
\end{proof}

We now introduce the \emph{adjacency} between shapes.
Intuitively, this concept captures how a single token jump connects different shapes.
Let $S_{1}$ and $S_{2}$ be sets having different shapes $\shape_{1}$ and $\shape_{2}$, respectively,
such that $S_{1} \setminus S_{2} = \{u\}$, $u \in V_{i}$,
$S_{2} \setminus S_{1} = \{v\}$, $v \in V_{j}$, and $i \ne j$.
For $h \in [t] \setminus \{i,j\}$, $\shape_{1}(h) = \shape_{2}(h)$ holds.
Since $\shape_{1} \ne \shape_{2}$, at least one of $\shape_{1}(i) \ne \shape_{2}(i)$ and $\shape_{1}(j) \ne \shape_{2}(j)$ holds.
If $\shape_{1}(i) \ne \shape_{2}(i)$, 
then $|V_{i} \cap S_{2}| = |V_{i} \cap S_{1}| - 1$ implies that one of the following holds:
\begin{itemize}
  \item A1: $\shape_{1}(i) \ne \bottop$, $\shape_{2}(i) \ne \bottop$, and $\shape_{2}(i) = \shape_{1}(i) - 1$;
  \item A2: $\shape_{1}(i) = \bottop$ and $\shape_{2}(i) = \numq(\phi) -1$;
  \hfill ($\signature_{S_{1}}(i)=\numq(\phi)$)
  \item A3: $\shape_{1}(i) = |V_{i}| - \numq(\phi) + 1$ and $\shape_{2}(i) = \bottop$.
  \hfill ($\signature_{S_{2}}(i)= |V_{i}| - \numq(\phi)$)
\end{itemize}
Similarly, if $\shape_{1}(j) \ne \shape_{2}(j)$, 
then $|V_{j} \cap S_{1}| = |V_{j} \cap S_{2}| - 1$ implies that one of the following holds:
\begin{itemize}
  \item B1: $\shape_{2}(j) \ne \bottop$, $\shape_{1}(j) \ne \bottop$, and $\shape_{1}(j) = \shape_{2}(j) - 1$;
  \item B2: $\shape_{2}(j) = \bottop$ and $\shape_{1}(j) = \numq(\phi) -1$;
  \hfill ($\signature_{S_{2}}(j)=\numq(\phi)$)
  \item B3: $\shape_{2}(j) = |V_{j}| - \numq(\phi) + 1$ and $\shape_{1}(j) = \bottop$.
  \hfill ($\signature_{S_{1}}(j)= |V_{j}| - \numq(\phi)$)
\end{itemize}

Given the observation above, we say that two $k$-feasible shapes 
$\shape_{1}$ and $\shape_{2}$ are \emph{adjacent} if and only if the following three conditions are satisfied.\\
(1) One of the following holds:
\begin{itemize}
  \item $\shape_{1}$ and $\shape_{2}$ disagree at exactly two indices $i$ and $j$
  such that $i$ satisfies one of A1, A2, A3
  and $j$ satisfies one of B1, B2, B3;

  \item $\shape_{1}$ and $\shape_{2}$ disagree at exactly one index $i$ satisfying one of A1, A2, A3,
  or $j$ satisfying one of B1, B2, B3.
\end{itemize}
(2) There exists a size-$k$ set $S_{1}$ of shape $\shape_{1}$ such that 
\begin{itemize}
  \item if $i$ is defined in (1) and $\shape_{1}(i) = \bottop$, then $\signature_{S_{1}}(i) = \numq(\phi)$;
  \item if $j$ is defined in (1) and $\shape_{1}(j) = \bottop$, then $\signature_{S_{1}}(j) = |V_{j}| - \numq(\phi)$.
\end{itemize}
(3) There exists a size-$k$ set $S_{2}$ of shape $\shape_{2}$ such that 
\begin{itemize}
  \item if $i$ is defined in (1) and $\shape_{2}(i) = \bottop$, then $\signature_{S_{2}}(i) = |V_{i}|-\numq(\phi)$;
  \item if $j$ is defined in (1) and $\shape_{2}(j) = \bottop$, then $\signature_{S_{2}}(j) = \numq(\phi)$.
\end{itemize}

The \emph{size-$k$ shape graph} $\mathcal{S}_{k}$ has 
the set of $k$-feasible shapes as its vertex set
and the adjacency between the vertices (shapes) is as defined above.
\begin{lemma}
\label{lem:nd-adjshape}
Let $\shape_{1}$ and $\shape_{2}$ be two different shapes that are $k$-feasible.
Then, $\shape_{1}$ and $\shape_{2}$ are adjacent in $\mathcal{S}_{k}$ if and only if
there exist size-$k$ feasible sets $S_{1}$ and $S_{2}$ of shapes $\shape_{1}$ and $\shape_{2}$, respectively,
with $|S_{1} \setminus S_{2}| = |S_{2} \setminus S_{1}| = 1$.
\end{lemma}
\begin{proof}

The if direction is already observed along the definition above.

To prove the only-if direction, assume that $\shape_{1}$ and $\shape_{2}$ are adjacent.
If $\shape_{1}$ and $\shape_{2}$ disagree at exactly two indices,
then we set $i$ to the one satisfying one of A1, A2, A3,
and $j$ to the one satisfying one of B1, B2, B3 in the definition.
If $\shape_{1}$ and $\shape_{2}$ disagree at exactly one index, then we set $i$ to this index.
Since $\shape_{1}$ and $\shape_{2}$ have the symmetric role, we can assume that $i$ satisfies one of A1, A2, A3 in the definition.

The definition of the shape adjacency implies the existence of a size-$k$ set $S_{1}$ of shape $\shape_{1}$
with the following additional conditions:
\begin{itemize}
  \item if $\shape_{1}(i) = \bottop$, then $\signature_{S_{1}}(i) = \numq(\phi)$;
  \item if the index $j$ is defined and $\shape_{1}(j) = \bottop$, then $\signature_{S_{1}}(j) = |V_{j}| - \numq(\phi)$.
\end{itemize}

If $\shape_{1}$ and $\shape_{2}$ disagree exactly at $i$,
and thus $j$ is not defined so far, we set $j$ to the index in the next claim.
\begin{claim}
\label{clm:nd-adjshape}
If $\shape_{1}$ and $\shape_{2}$ disagree exactly at $i$,
then there exists an index $j$ such that
$\shape_{1}(j) = \shape_{2}(j)$ ($= \bottop$) and $\signature_{S_{1}}(j) < |V_{j}| - \numq(\phi)$.
\end{claim}
\begin{proof}[Proof of \cref{clm:nd-adjshape}]
Let $S_{2}$ be a size-$k$ set of shape $\shape_{2}$.
Suppose to the contrary that for every $h \in [t] \setminus \{i\}$, 
it holds that $\shape_{1}(h) \ne \bottop$ or $\signature_{S_{1}}(h) = |V_{h}| - \numq(\phi)$.
Then, $|S_{1} \cap V_{h}| \ge |S_{2} \cap V_{h}|$ for all $h \in [t] \setminus \{i\}$.
Since the shapes $\shape_{1}$ and $\shape_{2}$ ensure that $|S_{1} \cap V_{i}| > |S_{2} \cap V_{i}|$,
we have $|S_{1}| > |S_{2}|$, a contradiction.
\end{proof}

Let $S_{2}$ be a set obtained from $S_{1}$ by removing a vertex of $V_{i}$ and adding a vertex of $V_{j}$.
The construction of $S_{1}$ implies that $S_{2}$ has shape $\shape_{2}$.
This completes the proof since $|S_{1} \setminus S_{2}| = |S_{2} \setminus S_{1}| = 1$.
\end{proof}

Since $|S| = |S'|= k$, the reachability between them can be reduced to 
the reachability between their shapes in $\mathcal{S}_{k}$.
\begin{lemma}
\label{lem:nd-shape-graph}
Let $\shape$ and $\shape'$ be the shapes of $S$ and $S'$, respectively.
There is a $\TJ(\phi)$-sequence from $S$ to $S'$
if and only if
$\shape$ and $\shape'$ belong to the same connected component of $\mathcal{S}_{k}$.
\end{lemma}
\begin{proof}
We first show the only-if direction. Assume that there is a $\TJ(\phi)$-sequence
$S_{0}, \dots, S_{\ell}$ from $S = S_{0}$ to $S' = S_{\ell}$.
For $0 \le i \le \ell$, let $\shape_{i}$ be the shape of $S_{i}$.
Since $S_{i}$ is a size-$k$ feasible set, $\shape_{i}$ belongs to $V(\mathcal{S}_{k})$.
By \cref{lem:nd-adjshape}, if $\shape_{i}$ and $\shape_{i+1}$ are different, then they are adjacent in $\mathcal{S}_{k}$.
Thus, $\mathcal{S}_{k}$ contains a walk from $\shape_{0} = \shape$ to $\shape_{\ell} = \shape'$.

To show the if direction, assume that there is a $\shape$--$\shape'$ path $\langle \shape_{0}, \dots, \shape_{\ell}\rangle$ in $\mathcal{S}_{k}$.
By \cref{lem:nd-adjshape},
for $0 \le i \le \ell-1$,
there exist size-$k$ sets $S'_{i}$ of shape $\shape_{i}$ and $S_{i+1}$ of shape $\shape_{i+1}$
with $|S'_{i} \setminus S_{i+1}| = |S_{i+1} \setminus S'_{i}| = 1$.
Furthermore, by \cref{lem:nd-sameshape}, for each $1 \le i \le \ell-1$,
there is a $\TJ(\phi)$-sequence from $S_{i}$ to $S'_{i}$.
Finally, observe that $S'_{0} = S$ and $S_{\ell} = S'$ by \cref{obs:nd-rshape}.
By combining these sequences, we obtain a $\TJ(\phi)$-sequence from $S$ to $S'$
because all intermediate sets are feasible by \cref{lem:nd-shape}.
\end{proof}

\cref{lem:nd-shape-graph} implies that \MSOR{1} can be solved by
checking that the shapes of the initial and target sets belong to the same connected component of $\mathcal{S}_{k}$.
We now show that $\mathcal{S}_{k}$ can be constructed efficiently.

\begin{lemma}
Constructing $\mathcal{S}_{k}$ can be done in time $\bigO(f(t + \numq(\phi)) \cdot n^{c})$
for some computable function $f$ and a constant $c$.
\end{lemma}
\begin{proof}
To enumerate the elements of $V(\mathcal{S}_{k})$ (i.e., the size-$k$ feasible shapes),
we first enumerate all $(2\numq(\phi)+1)^{t}$ shapes.
Then, for each shape $\shape$, we define $\shape_{\downarrow}$ and $\shape^{\uparrow}$ as follows:
\begin{align*}
  \shape_{\downarrow}(i) &= 
  \begin{cases}
    \numq(\phi) & \shape(i) = \bottop, \\
    \shape(i) & \text{otherwise},
  \end{cases}
  &
  \shape^{\uparrow}(i) &= 
  \begin{cases}
    |V_{i}| - \numq(\phi) & \shape(i) = \bottop, \\
    \shape(i) & \text{otherwise}.
  \end{cases}
\end{align*}
Observe that there is a set $S \subseteq V$ of size $k \le |V|$ with shape $\shape$
if and only if $\sum_{i\in [t]} \shape_{\downarrow}(i) \le k \le \sum_{i\in [t]} \shape^{\uparrow}(i)$.
If $\shape$ satisfies this condition, then by \cref{lem:nd-shape},
it suffices to check the feasibility of an arbitrary chosen set with shape $\shape$.
We construct such a set $S_{\shape}$ of shape $\shape$ by taking, say, $\shape_{\downarrow}(i)$ vertices from each $V_{i}$.
Testing whether $G \models \phi(S_{\shape})$ is fixed-parameter tractable parameterized by $\nd(G) + |\phi|$~\cite{CourcelleMR00,Lampis12}.

For each pair of vertices in $V(\mathcal{S}_{k})$,
the adjacency can be decided in $\bigO(t)$ time by checking each coordinate of the shapes
and testing the existence of size-$k$ sets of the given shapes with the additional conditions.
\end{proof}

The lemma above already implies that \MSOR{1} is fixed-parameter tractable parameterized by $\nd + |\phi|$.
To find a shortest $\TJ(\phi)$-sequence, we take a closer look at~$\mathcal{S}_{k}$.

A sequence $\shape_{0}, \dots, \shape_{q}$ of shapes with $\shape_{i} \ne \shape_{i+1}$ for $0 \le i < q $ is 
the \emph{shape sequence} of a $\TJ(\phi)$-sequence
if the $\TJ(\phi)$-sequence can be split into $q+1$ subsequences
such that all sets in the $i$th subsequence have shape $\shape_{i}$ for $0 \le i \le q$.

\begin{lemma}
\label{lem:nd-path}
If there is a $\TJ(\phi)$-sequence from $S$ to $S'$,
then there is a shortest one such that the corresponding shape sequence
forms a simple path in $\mathcal{S}_{k}$.
\end{lemma}
\begin{proof}
Let $S_{0}, \dots, S_{p}$ be a shortest $\TJ(\phi)$-sequence from $S$ to $S'$
and $\shape_{0}, \dots, \shape_{q}$ be its shape sequence.
By \cref{lem:nd-adjshape}, $\shape_{0}, \dots, \shape_{q}$ is a walk in $\mathcal{S}_{k}$.
Assume that $S_{0}, \dots, S_{p}$ has minimum $q$ among all shortest $\TJ(\phi)$-sequences from $S$ to $S'$.
Suppose that $\shape_{i} = \shape_{j}$ for some $i < j$.
We take $i$ and $j$ so that $j-i$ is maximized.
Let $h_{i}$ and $h_{j}$ be the smallest and largest indices such that 
$S_{h_{i}}$ and $S_{h_{j}}$ have the shape $\shape_{i}$ ($= \shape_{j}$).
By \cref{lem:nd-sameshape}, there is a $\TJ(\phi)$-sequence $R_{0}, \dots, R_{|S_{h_{i}} \setminus S_{h_{j}}|}$
of length $|S_{h_{i}} \setminus S_{h_{j}}|$ from $S_{h_{i}} = R_{0}$ to $S_{h_{j}} = R_{|S_{h_{i}} \setminus S_{h_{j}}|}$
such that all sets in the sequence have the same shape.
Since every $\TJ(\phi)$-sequence from $S_{h_{i}}$ to $S_{h_{j}}$ takes at least $|S_{h_{i}} \setminus S_{h_{j}}|$ steps,
$S_{0}, \dots, R_{0} = (S_{h_{i}}), \dots, R_{|S_{h_{i}} \setminus S_{h_{j}}|} (= S_{h_{j}}), \dots, S_{p}$
is also a shortest $\TJ(\phi)$-sequence from $S$ to $S'$.
This new sequence has $\shape_{0}, \dots, \shape_{i}, \shape_{j+1}, \dots, \shape_{q}$ as its shape sequence.
This contradicts the assumption that $q$ is minimum.
\end{proof}

\cref{lem:nd-path} implies that for finding a shortest $\TJ(\phi)$-sequence from $S$ to $S'$,
it suffices to first guess a path in $\mathcal{S}_{k}$ and then find a shortest $\TJ(\phi)$-sequence having the path as its shape sequence.
Note that $|V(\mathcal{S}_{k})| \le (2\numq(\phi)+1)^{t}$ and thus the number of candidates for such shape sequences
is upper bounded by a function depending only on $\numq(\phi)$ and $t$. (Recall that $t$ is the number of types in $G$.)
Therefore, the following lemma completes the proof of \cref{thm:msor1-nd}.

\begin{lemma}
\label{lem:nd-shortest}
Given a sequence $\shape_{0}, \dots, \shape_{q}$ of shapes such that $\shape_{0} = \shape_{S}$ and $\shape_{q} = \shape_{S'}$,
finding a shortest $\TJ(\phi)$-sequence with the shape sequence $\shape_{0}, \dots, \shape_{q}$ is
fixed-parameter tractable parameterized by $t + \numq(\phi)$.
\end{lemma}
\begin{proof}
We reduce the problem to \textsc{Minimum-Cost Circulation} defined as follows.
Let $D = (X,A)$ be a directed graph.
We define $\delta^{\mathrm{in}}(v) = \{a \in A \mid a = (u,v) \in A\}$ and
$\delta^{\mathrm{out}}(v) = \{a \in A \mid a = (v,u) \in A\}$.
A function $f \colon A \to \mathbb{R}$ is a \emph{circulation}
if $f(\delta^{\mathrm{in}}(v)) = f(\delta^{\mathrm{out}}(v))$ for each $v \in X$,
where $f(A') = \sum_{a \in A'} f(a)$ for $A' \subseteq A$.
A circulation $f$ is an \emph{integer circulation} if $f(a)$ is an integer for each $a \in A$.
Given a \emph{cost function} $w \colon A \to \mathbb{Q}$,
the \emph{cost} of a circulation $f$ is defined as $\mathrm{cost}(f) = \sum_{a \in A} w(a) f(a)$.
Now, given a directed graph $D = (X,A)$, 
a \emph{demand} function $d \colon A \to \mathbb{Q}$,
a \emph{capacity} function $c \colon A \to \mathbb{Q}$, and
a cost function $w \colon A \to \mathbb{Q}$,
\textsc{Minimum-Cost Circulation} asks to find a circulation $f$
minimizing $\mathrm{cost}(f)$ under the condition that $d(a) \le f(a) \le c(a)$ for each $a \in A$.
It is known that 
\textsc{Minimum-Cost Circulation} can be solved in strongly polynomial time, and
if the demand $d$ and the capacity $c$ take integer values only,
then a minimum-cost integer circulation is found~\cite{Tardos85}.

Now we construct an instance of \textsc{Minimum-Cost Circulation}
from the graph $G = (V, E, \mathcal{C})$,
its type partition $\langle V_{1}, \dots, V_{t} \rangle$,
and the shape sequence $\shape_{0}, \dots, \shape_{q}$.
We first construct $D = (X,A)$.
The digraph $D$ contains two special vertices $s$ and $s'$,
and $q$ sets $L_{0}, L_{1}, \dots, L_{q-1}$ of vertices
such that $L_{j} = \{v_{1}^{j}, \dots, v_{t}^{j}\}$ for $0 \le j \le q-1$.
Each $L_{j}$ is a bidirectional clique (i.e., there is an arc for each ordered pair of vertices in $L_{j}$).
For $1 \le j \le q-1$, $D$ contains the matching $\{(v_{i}^{j-1}, v_{i}^{j}) \mid 1 \le i \le t\}$ from $L_{j-1}$ to $L_{j}$.
There are arcs from $s$ to all vertices in $L_{0}$
and from all vertices in $L_{q-1}$ to $s'$.
Additionally, $D$ contains the arc $(s',s)$.
Each arc $a$ in each clique $L_{j}$ has demand $d(a) = 0$, capacity $c(a) = \infty$, and cost $w(a) = 1$.
All other arcs have cost $0$.
We set $d((s',s)) = c((s',s)) = k$.
For $i \in [t]$, we set 
$d((s, v_{i}^{0})) = c((s, v_{i}^{0})) = |V_{i} \cap S|$ ($= \shape_{0}(i)$) and
$d((v_{i}^{q-1}, s')) = c((v_{i}^{q-1}, s')) = |V_{i} \cap S'|$ ($= \shape_{q}(i)$).
For $i \in [t]$ and $j \in [q-1]$, we set
\begin{align*}
  d((v_{i}^{j-1}, v_{i}^{j}))
  &=
  \begin{cases}
    \numq(\phi) &  \shape_{j}(i) = \bottop, \\
    \shape_{j}(i) & \text{otherwise},
  \end{cases}
  &
  c((v_{i}^{j-1}, v_{i}^{j}))
  &=
  \begin{cases}
    |V_{i}| - \numq(\phi) &  \shape_{j}(i) = \bottop, \\
    \shape_{j}(i) & \text{otherwise}.
  \end{cases}
\end{align*}

We show that 
there exists a $\TJ(\phi)$-sequence of length at most $p$ with the shape sequence $\shape_{0}, \dots, \shape_{q}$ from $S$ to $S'$
if and only if 
the instance $\langle D, d, c, w \rangle$ of \textsc{Minimum-Cost Circulation}
admits an integer circulation $f$ of cost at most $p$.
This completes the proof
since \textsc{Minimum-Cost Circulation} is solvable in strongly polynomial time
and the size of $D$ depends only on $t$ and the number of shapes.

\paragraph{The only-if direction}
Assume that 
$S_{0}, \dots, S_{p}$ is a $\TJ(\phi)$-sequence from $S = S_{0}$ to $S' = S_{p}$ with the shape sequence $\shape_{0}, \dots, \shape_{q}$.
Let $h(0), \dots, h(q)$ be the first indices such that $S_{h(j)}$ has shape $\shape_{j}$.
Note that $h(0) = 0$ and $h(q) = p$ by \cref{obs:nd-rshape}.
We construct a circulation $f$ as follows.

For $i \in [t]$ and $j \in [q-1]$, we set
$f((v_{i}^{j-1}, v_{i}^{j})) = |V_{i} \cap S_{h(j)}|$.
Since $S_{h(j)}$ has shape $\shape_{j}$, we have
$d((v_{i}^{j-1}, v_{i}^{j})) \le f((v_{i}^{j-1}, v_{i}^{j})) \le c((v_{i}^{j-1}, v_{i}^{j}))$.
For $i,i' \in [t]$ with $i \ne i'$ and for $0 \le j \le q - 1$,
we set $f((v_{i}^{j}, v_{i'}^{j}))$ to the number of $\TJ(\phi)$-moves from $V_{i}$ to $V_{i'}$
in the $\TJ(\phi)$-sequence $S_{h(j)}, \dots, S_{h(j+1)}$.
For the other arcs $a$, we set $f(a) = d(a)$ ($=c(a)$).

To see that $f$ is a circulation, 
we need to ensure that $f(\delta^{\mathrm{in}}(v)) = f(\delta^{\mathrm{out}}(v))$ for each $v \in X$.
This holds for the special vertices $s$ and $s'$ since each arc $a$ incident to them satisfies that $d(a) = c(a)$.
For $v_{i}^{j}$, observe that it has
an incoming arc $a_{\text{in}}$ from $s$ or $v_{i}^{j-1}$ such that $f(a_{\text{in}}) = |V_{i} \cap S_{h(j)}|$
and
an outgoing arc $a_{\text{out}}$ to $s'$ or $v_{i}^{j+1}$ such that $f(a_{\text{out}}) = |V_{i} \cap S_{h(j+1)}|$.
Thus we have
\begin{align*}
  f(\delta^{\mathrm{in}}(v_{i}^{j}))
  &= 
  |V_{i} \cap S_{h(j)}| + \sum_{i' \in [t] \setminus \{i\}} f((v_{i'}^{j}, v_{i}^{j})),
  \\
  f(\delta^{\mathrm{out}}(v_{i}^{j}))
  &= 
  |V_{i} \cap S_{h(j+1)}| + \sum_{i' \in [t] \setminus \{i\}} f((v_{i}^{j}, v_{i'}^{j})).
\end{align*}
On the other hand, the definition of $f$ in $L_{j}$ implies that 
\[
  |V_{i} \cap S_{h(j)}| + \sum_{i' \in [t] \setminus \{i\}} f((v_{i'}^{j}, v_{i}^{j})) - \sum_{i' \in [t] \setminus \{i\}} f((v_{i}^{j}, v_{i'}^{j}))
  =
  |V_{i} \cap S_{h(j+1)}|.
\]
Hence, $f(\delta^{\mathrm{in}}(v_{i}^{j})) = f(\delta^{\mathrm{out}}(v_{i}^{j}))$ holds.

To see that $\mathrm{cost}(f) = p$,
recall that only the arcs in the cliques $L_{j}$ have positive costs
and that each step in the $\TJ(\phi)$-sequence $S_{0}, \dots, S_{p}$ contributes exactly $1$ to $f$ in a clique $L_{j}$.

\paragraph{The if direction}
Assume that $\langle D, d, c, w \rangle$ admits an integer circulation $f$ of cost at most~$p$.
We assume that $f$ is of minimum cost. 

Let $j$ be an index such that $0 \le j \le q-1$.
Let $R = S$ if $j = 0$; otherwise let $R$ be a set of shape $\shape_{j}$ such that 
$|V_{i} \cap R| = f((v_{i}^{j-1},v_{i}^{j}))$ for each $i \in [t]$.
We show that there is a set $R'$ of shape $\shape_{j+1}$ such that 
\begin{itemize}
  \item $R' = S'$ if $j+1 = q$, otherwise $|V_{i} \cap R'| = f((v_{i}^{j},v_{i}^{j+1}))$ for each $i \in [t]$, and
  \item there is a $\TJ(\phi)$-sequence $\langle R_{0}, \dots, R_{\ell'}, \dots, R_{\ell} \rangle$ with $\ell = \sum_{a \in A(D[L_{j}])} f(a)$
  from $R_{0} = R$ to $R_{\ell} = R'$, where 
  each $R_{i}$ with $0 \le i \le \ell'$ has shape $\shape_{j}$ and
  each $R_{i}$ with $\ell'+1 \le i \le \ell$ has shape $\shape_{j+1}$.
\end{itemize}
Concatenating such sequences, we can obtain a $\TJ(\phi)$-sequence from a set of shape $\shape_{0}$
to a set of shape $\shape_{q}$ (i.e., from $S$ to $S'$ by \cref{obs:nd-rshape})
with length 
\[
  \sum_{j=0}^{q-1} \, \sum_{a \in A(D[L_{j}])} f(a) \le p.
\]

Observe that in $D[L_{j}]$, no vertex $v_{i}^{j}$ has both outgoing and incoming arcs with positive $f$ values.
To see this, suppose that $f((v_{i}^{j}, v_{i'}^{j})) = a > 0$ and
$f((v_{i''}^{j}, v_{i}^{j})) = b > 0$ for some $i'$ and $i''$.
Then, by subtracting $\min\{a,b\}$ from both $f((v_{i}^{j}, v_{i'}^{j}))$ and $f((v_{i''}^{j}, v_{i}^{j}))$,
and by adding $\min\{a,b\}$ to $f((v_{i''}^{j}, v_{i'}^{j}))$, we obtain a circulation with strictly smaller cost.
This contradicts the assumption that $f$ is of minimum cost.
Let $L_{j}^{\text{out}}$ and $L_{j}^{\text{in}}$
be the sets of vertices in $L_{j}$ with outgoing arcs and incoming arcs in $D[L_{j}]$, respectively, with positive $f$ values.

We initialize a mapping $g$ with the restriction of $f$ to $A(D[L_{j}])$.
Starting with $R_{0} = R$, we obtain $R_{i+1}$ from $R_{i}$
by picking up some arc $(v_{i}^{j}, v_{i'}^{j})$ with $g((v_{i}^{j}, v_{i'}^{j})) > 0$,
making a $\TJ(\phi)$-move from $V_{i}$ to $V_{i'}$,
and updating $g$ by subtracting $1$ from $g((v_{i}^{j}, v_{i'}^{j}))$.
After executing such steps $\sum_{a \in A(D[L_{j}])} f(a)$ times, we obtain $R'$.
Thus it suffices to show that we can execute such steps in a right ordering
so that the requirements for the shapes are satisfied.
Observe that for each $i$, $|R_{x} \cap V_{i}|$ changes monotonically for $0 \le x \le \ell$:
it is decreasing if $v_{i}^{j} \in L_{j}^{\text{out}}$,
increasing if $v_{i}^{j} \in L_{j}^{\text{in}}$, and
it never changes if $v_{i}^{j} \notin L_{j}^{\text{out}} \cup L_{j}^{\text{in}}$.
Thus, if $\shape_{j}(i) = \shape_{j+1}(i)$, then for each $R_{x}$ in the $\TJ(\phi)$-sequence,
it holds that $\shape_{R_{x}}(i) = \shape_{j}(i)$ ($=\shape_{j+1}(i)$), where $\shape_{R_{x}}$ is the shape of $R_{x}$.
Thus, when constructing the $\TJ(\phi)$-sequence, we only have to take care of indices $h$ with $\shape_{j}(h) \ne \shape_{j+1}(h)$.

First assume that there is only one index $h$ with $\shape_{j}(h) \ne \shape_{j+1}(h)$.
In this case, we can make the moves in any order.
Only the $h$th component of the shape changes during the $\TJ(\phi)$-sequence,
and when it changes, the shape becomes equal to $\shape_{j+1}$ and never changes after that.

Next assume that $\shape_{j}$ and $\shape_{j+1}$ disagree at exactly two indices.
We set $i^{\text{out}}$ to the index satisfying one of A1, A2, A3 in the definition of adjacency
and $i^{\text{in}}$ to the one satisfying one of B1, B2, B3.
From the definition of adjacency, $v_{i^{\text{out}}}^{j} \in L_{j}^{\text{out}}$ and $v_{i^{\text{in}}}^{j} \in L_{j}^{\text{in}}$.
We first execute moves that do no change the shape.
After exhaustively executing such moves, we obtain a set $R_{x}$ with shape $\shape_{j}$ satisfying that 
\begin{align*}
  |V_{i^{\text{out}}} \cap R_{x}| 
  &= 
  \begin{cases}
    \shape_{j+1}(i^{\text{out}}) + 1 & \shape_{j+1}(i^{\text{out}}) \ne \bottop, \\
    |V_{i^{\text{out}}}| - \numq(\phi) + 1 & \shape_{j+1}(i^{\text{out}}) = \bottop,
  \end{cases}
  \\
  |V_{i^{\text{in}}} \cap R_{x}| 
  &= 
  \begin{cases}
    \shape_{j+1}(i^{\text{in}}) - 1 & \shape_{j+1}(i^{\text{in}}) \ne \bottop, \\
    \numq(\phi)-1 & \shape_{j+1}(i^{\text{in}}) = \bottop.
  \end{cases}
\end{align*}
Observe that at this point, each positive arc in $g$ is
an outgoing arc of $v_{i^{\text{out}}}^{j}$ or an incoming arc of $v_{i^{\text{in}}}^{j}$.

Now we make a move from $V_{i^{\text{out}}}$ to $V_{i^{\text{in}}}$ and obtain a set $R_{x+1}$ with shape $\shape_{j+1}$
no matter if $g((v_{i^{\text{out}}}^{j}, v_{i^{\text{in}}}^{j})) > 0$ or not.
If $g((v_{i^{\text{out}}}^{j}, v_{i^{\text{in}}}^{j})) > 0$,
then this is just a valid move. We decrease $g((v_{i^{\text{out}}}^{j}, v_{i^{\text{in}}}^{j}))$ by $1$ in this case.
If $g((v_{i^{\text{out}}}^{j}, v_{i^{\text{in}}}^{j})) = 0$,
then there exist $h^{\text{out}}$ and $h^{\text{in}}$ such that
$g((v_{i^{\text{out}}}^{j}, v_{h^{\text{in}}}^{j})) > 0$ and $g((v_{h^{\text{out}}}^{j}, v_{i^{\text{in}}}^{j})) > 0$
as moves from $V_{i^{\text{out}}}$ and to $V_{i^{\text{in}}}$ still have to be made.
We decrease both
$g((v_{i^{\text{out}}}^{j}, v_{h^{\text{in}}}^{j}))$ and $g((v_{h^{\text{out}}}^{j}, v_{i^{\text{in}}}^{j}))$ by $1$
and increase $g((v_{h^{\text{out}}}^{j}, v_{h^{\text{in}}}^{j}))$ by $1$.
Intuitively, we make a move from $V_{i^{\text{out}}}$ to $V_{i^{\text{in}}}$ and promise to make a move from $V_{h^{\text{out}}}$ to $V_{h^{\text{in}}}$
instead of making moves from $V_{i^{\text{out}}}$ to $V_{h^{\text{in}}}$ and from $V_{h^{\text{out}}}$ to $V_{i^{\text{in}}}$.

After obtaining the set $R_{x+1}$ with shape $\shape_{j+1}$, we execute moves following $g$ in an arbitrary order.
This never changes the shape of the set and obtains $R'$.
\end{proof}

\subsection{Extension to \MSOR{2}}
\label{mso2ND}
Here we try to extend \cref{thm:msor1-nd} to \MSOR{2}.
We first observe that it is easy if the parameter is the vertex cover number.
\begin{proposition}
[{\cite[Lemma~6 (rephrased)]{Lampis12}}]
\label{pro:vc-mso2}
Given an \mso{2} sentence $\phi$, a colored graph $G$, and a vertex cover $C$ of $G$,
one can compute in polynomial time
an \mso{1} sentence $\phi'$ and a colored graph $G'$ with $|C|$ additional colors such that
$\vc(G') = \vc(G)$,
$|\phi'|$ depends only on $|\phi|$ and $|C|$,  and
$G \models \phi$ if and only if $G' \models \phi'$.
\end{proposition}

One can easily extend \cref{pro:vc-mso2} to the corresponding proposition about 
an \mso{2} formula with one free vertex-set variable $X$ by considering $X$ as a new color.
This observation and \cref{thm:msor1-nd} together imply the following.
(Recall that $\nd(G) \le 2^{\vc(G)} + \vc(G)$ for every graph $G$.)
\begin{corollary}
\label{cor:msor2-vc}
\MSOR{2} parameterized by $\vc + |\phi|$ is fixed-parameter tractable.
Furthermore, for a yes instance of \MSOR{2}, finding a shortest $\TJ$-sequence is fixed-parameter tractable with the same parameter.
\end{corollary}

Unfortunately, such an extension of \cref{thm:msor1-nd} with its full generality is not possible under some reasonable assumption.
We use the following hardness result on \mso{2} \textsc{Model Checking} on complete graphs.
Given an \mso{2} sentence $\psi$ and a graph $G$,
\mso{2} \textsc{Model Checking} asks whether $G \models \psi$.
\begin{proposition}
[\cite{Lampis14}]
Unless $\mathrm{E} = \mathrm{NE}$,\footnote{%
Recall that $\mathrm{E} = \mathrm{DTIME}(2^{\bigO(n)})$ and $\mathrm{NE} = \mathrm{NTIME}(2^{\bigO(n)})$.
The same hardness result with a stronger assumption $\mathrm{EXP} \ne \mathrm{NEXP}$ was shown earlier by 
Courcelle, Makowsky, and Rotics~\cite{CourcelleMR00}.} 
\mso{2} \textsc{Model Checking} on $n$-vertex uncolored complete graphs
cannot be solved in time $\bigO(n^{f(|\phi|)})$ for any function~$f$.
\end{proposition}

\begin{theorem}
\label{thm:nd-msor2}
Unless $\mathrm{E} = \mathrm{NE}$, 
\MSOR{2} on $n$-vertex uncolored graphs of neighborhood diversity~$2$
cannot be solved in time $\bigO(n^{f(|\phi|)})$ for any function~$f$.
\end{theorem}
\begin{proof}
Let $\psi$ be an \mso{2}-sentence and $K_{n}$ be an uncolored complete graph of $n$ vertices.
We construct an instance $\langle \phi, G, S, S' \rangle$ of \MSOR{2} as follows.
The graph $G$ is obtained from $K_{n+4}$ by removing the three edges in a triangle formed by $x,y,z \in V(K_{n+4})$.
Let $a$ and $b$ be arbitrary two vertices in $V(G) \setminus \{x,y,z\}$.
We set $S = V(G) \setminus \{y,z\}$ and $S' = V(G) \setminus \{a,b\}$.
To define the formula $\phi$, we first modify the sentence $\psi$ by 
adding a new free vertex-set variable $X$,
asking all vertex variables and vertex-set variables in $\psi$
to be elements and subsets of $X$, respectively, and 
asking all edge variables and edge-set variables in $\psi$
to be elements and subsets of $E(G[X])$, respectively.
We call the obtained \mso{2} formula $\psi'$.
Note that \cref{prop:isomorphic}
implies that for every $n$-vertex clique $X$ of $G$,
$\psi'(X)$ is true if and only if $K_{n} \models \psi$.
We also need the following formula expressing that 
the vertex set $G[X]$ has all but one possible edges (i.e., $X$ is an \emph{almost-clique}):
\[
  \textsf{almost-clique}(X) = \dot{\exists} u, v \in X \colon
  \lnot E(u,v) \land 
  (\dot{\forall} p, q \in X \colon \lnot E(p,q) \Rightarrow \{p,q\} = \{u,v\}),
\]
where $\{p,q\} = \{u,v\}$ is a syntax sugar for $(p=u \land q = v) \lor (p=v \land q = u)$.
Now we define $\phi$ as follows:
\[
  \phi(X) = \textsf{almost-clique}(X) \Rightarrow 
  \dot{\exists} u,v \in X \colon \lnot E(u,v) \land \psi'(X \setminus \{u,v\}),
\]
where $\psi'(X \setminus \{u,v\})$
can be expressed as
$\exists Y \subseteq X \colon \psi'(Y) \land (\forall w \in X \colon \lnot w \in Y \rightarrow (w = u) \lor (w = v))$.
Note that $G \models \phi(S)$ and $G \models \phi(S')$
since $S$ is a clique and $S'$ misses three edges.

We show that $K_{n} \models \psi$
if and only if $\langle \phi, G, S, S' \rangle$ is a yes-instance of \MSOR{2}.
This implies the theorem as $\nd(G) = 2$ ($\{x,y,z\}$ and $V(G) \setminus \{x,y,z\}$ are the twin classes).

To show the only-if direction, assume that $K_{n} \models \psi$.
We set $S'' = S-a+y$ ($= S' + b -z$).
Observe that $S''$ is an almost-clique with the missing edge $\{x,y\}$
and $S'' \setminus \{x,y\}$ is a clique of size $n$.
Hence, $G \models \phi(S'')$ holds.
Therefore, $\langle S, S'', S' \rangle$ is a $\TJ(\phi)$-sequence.

To show the if direction, assume that there is a $\TJ(\phi)$-sequence from $S$ to $S'$.
Since $|S \cap \{x,y,z\}| = 1$ and $|S' \cap \{x,y,z\}| = 3$,
the $\TJ(\phi)$-sequence contains a set $S''$ with $|S'' \cap \{x,y,z\}| = 2$.
Observe that the set $S''$ is an almost-clique and $S'' \setminus \{x,y,z\}$ is a clique of size $n$.
Since $G \models \phi(S'')$, we have $G \models \psi'(S'' \setminus \{x,y,z\})$,
and thus $K_{n} \models \psi$.
\end{proof}

The proof of \cref{thm:nd-msor2} also shows that \MSOR{2} is hard for graphs of twin cover number~$3$.
A vertex set $C \subseteq V$ of a graph $G = (V,E)$ is a \emph{twin cover}~\cite{Ganian11}
if $V \setminus C$ is a disjoint union of cliques such that each clique is a set of twins in $G$.
The \emph{twin cover number} of a graph is the minimum size of a twin cover in the graph.
The graph $G$ in the proof has a twin cover $\{x,y,z\}$,
and thus \cref{thm:nd-msor2} implies the following corollary.
\begin{corollary}
\label{cor:tc-msor2}
Unless $\mathrm{E} = \mathrm{NE}$, 
\MSOR{2} on $n$-vertex uncolored graphs of twin cover number~$3$
cannot be solved in time $\bigO(n^{f(|\phi|)})$ for any function $f$.
\end{corollary}


\section{Fixed-parameter algorithm parameterized by the solution size and treedepth}
\label{sec:fpt}

In this section, we show that \MSOR{2} is fixed-parameter tractable
when parameterized simultaneously by treedepth, the length of the \mso{2} formula, and the size of input sets $S$ and $S'$.

\begin{theorem}
\label{thm:msor2-td+k}
\MSOR{2} parameterized by $\td + k + |\phi|$ is fixed-parameter tractable, where $k$ is the size of input sets.
Furthermore, for a yes instance of \MSOR{2}, finding a shortest $\TJ$-sequence is fixed-parameter tractable with the same parameter.
\end{theorem}
As we show in \cref{sec:pspace-c},
having the size of input sets is necessary since otherwise it is PSPACE-complete.

It is known (see e.g.,~\cite{CourcelleE12}) that given a colored graph $G$ and an \mso{2} sentence $\phi$,
one can compute in polynomial time a colored graph $G'$ and an \mso{1} sentence $\phi'$
such that
\begin{itemize}
  \item $G \models \phi$ if and only if $G' \models \phi'$;
  \item $G'$ is obtained from $G$ by subdividing each edge,
  and consider the set of new vertices introduced by the subdivisions as a new color;
  \item the length of $\phi'$ is bounded by a function of $|\phi|$.
\end{itemize}
Observe that $\td(G') \le \td(G) +1$.\footnote{%
Starting with a treedepth decomposition $F$ of $G$ with depth at most $d$,
we construct a treedepth decomposition $F'$ of $G'$ with depth at most $d+1$
by adding the vertex corresponding to each edge $\{u,v\} \in E(G)$
as a leaf attached to one of $u$ and $v$ that is a descendant of the other.}
Thus, to prove \cref{thm:msor2-td+k}, it suffices to 
show that \MSOR{1} is fixed-parameter tractable parameterized by the claimed parameter.

Now we generalize the type of a vertex used in \cref{sec:vc-nd} to the type of a vertex set.
For a colored graph $G = (V,E,\mathcal{C})$ and vertex sets $X, X' \subseteq V$,
we say that $X$ and $X'$ have the same \emph{type} if there is an isomorphism $\eta$ from $G$ to itself
such that 
$\eta(X) = X'$, $\eta(X') = X$, and
$\eta(v) = v$ for every $v \notin X \cup X'$.
Note that from the definition of isomorphisms between colored graphs,
$\mathcal{C}(v) = \mathcal{C}(\eta(v))$ holds for every $v \in V$.
Note also that singletons $\{x\}, \{x'\} \subseteq V$ have the same type
if and only if the vertices $x$ and $x'$ have the same type.

The next lemma says that if there are many disjoint vertex sets of the same type,
then we can avoid most of them when finding $\TJ(\phi)$-sequences.
\begin{lemma}
\label{lem:only-k-sets}
Let $\langle \phi, G, S, S' \rangle$ be a yes-instance of \MSOR{1} with $|S| = |S'| = k$.
Let $C_{1},\dots,C_{t}$ be a family of disjoint vertex sets with the same type
not intersecting $S \cup S'$.
If $t > k$, then for every $I \subseteq [t]$ with $|I| = k$,
there is a shortest $\TJ(\phi)$-sequence $S_{0},\dots,S_{\ell}$ from $S_{0} = S$ to $S_{\ell} = S'$
such that 
$C_{i} \cap \bigcup_{0 \le j \le \ell} S_{j} \ne \emptyset$ only if $i \in I$.
\end{lemma}
\begin{proof}
Without loss of generality, assume that $I = [k]$.
Let $S_{0},\dots,S_{\ell}$ be a shortest $\TJ(\phi)$-sequence from $S_{0} = S$ to $S_{\ell} = S'$.
If this sequence has no intersection with sets $C_{k+1}, \dots, C_{t}$, we are done.
Assume that $r$ is the first index such that $C_{p} \cap S_{r} \ne \emptyset$ for some $p \notin [k]$.
Note that $p$ is unique as a $\TJ(\phi)$-move adds only one new vertex.
We assume that the shortest $\TJ(\phi)$-sequence is chosen so that the index $r$ is maximized.
Observe that $r \ge 1$ as $S_{0}$ intersects no $C_{i}$.

Now we construct a new $\TJ(\phi)$-sequence $S_{0},\dots,S_{r-1}, S_{r}', \dots, S'_{\ell}$ from $S_{0} = S$ to $S'_{\ell} = S'$.
Since $|S_{r}| = k$, there exists an index $q \in [k]$ such that $C_{q} \cap S_{r} = \emptyset$.
We name the vertices in $C_{p}$ and $C_{q}$ as $C_{p} = \{u_{1},\dots,u_{c}\}$ and $C_{q} = \{v_{1},\dots,v_{c}\}$
so that there is an isomorphism $f$ from $G$ to itself
that maps $u_{i}$ to $v_{i}$ for $i \in [c]$, $v_{i}$ to $u_{i}$ for $i \in [c]$, and the other vertices to themselves.
This is possible since $C_{p}$ and $C_{q}$ have the same type. 
For $r \le j \le \ell$, we define $S'_{j}$ by \emph{swapping} $C_{p}$ and $C_{q}$ as follows:
\[
 S'_{j} = (S_{j} \setminus (C_{q} \cup C_{p})) \cup \{u_{i} \mid v_{i} \in S_{j}\} \cup \{v_{i} \mid u_{i} \in S_{j}\}.
\]
Note that the assumption $S_{\ell} \cap (C_{q} \cup C_{p}) = \emptyset$ implies that $S'_{\ell} = S_{\ell}$.
By the assumption $G \models \phi(S_{j})$, the existence of $f$, and \cref{prop:isomorphic},
it holds that $G \models \phi(S'_{j})$ for $r \le j \le \ell$.
Thus the sequence $S'_{r},\dots,S'_{\ell}$ is a $\TJ(\phi)$-sequence.

We now show that $|S_{r-1} \setminus S'_{r}| = |S'_{r} \setminus S_{r-1}| = 1$.
If $|S_{r-1} \cap C_{q}| = 1$, then $S_{r}$ is the set obtained from $S_{r-1}$ by removing the unique element in $C_{q}$
and adding some element in $C_{p}$ for the first time. In this case, $S_{r-1} = S'_{r}$ holds,
and thus $S_{0},\dots,S_{r-1}, S'_{r+1}, \dots, S'_{\ell}$ is a $\TJ(\phi)$-sequence from $S$ to $S'$.
This contradicts the assumption that $S_{0},\dots,S_{\ell}$ is a shortest one.
Since $S_{r} \cap C_{q} = \emptyset$, we have $S_{r-1} \cap C_{q} = \emptyset$.
Let $S_{r-1} \setminus S_{r} = \{x\}$ and $S_{r} \setminus S_{r-1} = \{u_{i}\}$.
Since $x \notin C_{p} \cup C_{q}$,
we have that $S_{r-1} \setminus S'_{r} = \{x\}$ and $S'_{r} \setminus S_{r-1} = \{v_{i}\}$ as desired.

The discussions so far show that $S_{0},\dots,S_{r-1}, S'_{r}, \dots, S'_{\ell}$ is a (shortest) $\TJ(\phi)$-sequence from $S$ to $S'$.
However, the first index $r'$ (if any exists) such that $S'_{r'} \cap C_{i}$ for some $i \notin [k]$ is larger than $r$
as $C_{i} \cap S'_{r} = \emptyset$ for every $i \notin [k]$.
This contradicts the assumption on $r$.
\end{proof}

Next we further argue that if there are a much larger number of disjoint vertex sets of the same type, 
then we can safely remove some of them. 
Note that this claim is stronger than \cref{lem:only-k-sets} in some sense.
Since the formula $\phi(X)$ may depend on the whole structure of $G$ (i.e., not only on $G[X]$),
``not using it in a sequence'' and ``removing it from the graph'' are different.

We need the following proposition, which is a generalization of \cref{prop:Lampis12}.
\begin{proposition}
[\cite{LampisM21}]
\label{prop:many-sets}
Let $G$ be a colored graph and $\phi$ be an \mso{1} formula with one free set variable.
Assume that $G$ contains $t > 2^{p \cdot \numq(\phi)}$
disjoint size-$p$ vertex sets with the same type.
Let $G'$ be the graph obtained from $G$ by removing one of the $t$ sets.
Then, for every subset $X \subseteq V$ disjoint from the $t$ sets,
$G' \models \phi(X)$ if and only if $G \models \phi(X)$.
\end{proposition}

The proposition above was originally stated for \mso{1} sentences in \cite{LampisM21}
but can be easily modified to this form by considering the free set variable as a new color.

Now we can prove the key lemma.
\begin{lemma}
\label{lem:removing-a-set}
Let $\langle \phi, G, S, S' \rangle$ be an instance of \MSOR{1} with $|S| = |S'| = k$.
Let $C_{1},\dots,C_{t}$ be a family of disjoint size-$p$ vertex sets with the same type
not intersecting $S \cup S'$.
If $t > k + 2^{p \cdot \numq(\phi)}$, then for every $C \in \{C_{1},\dots,C_{t}\}$,
$\dist_{\phi,G}(S,S') = \dist_{\phi,G-C}(S,S')$.
\end{lemma}
\begin{proof}
Without loss of generality, assume that $C = C_{t}$.
Let $G' = G - C_{t}$.
By \cref{lem:only-k-sets}, we can focus on shortest $\TJ(\phi)$-sequences (both in $G$ and in $G'$)
that may intersect $C_{1}, \dots, C_{k}$ but do not intersect $C_{k+1}, \dots, C_{t}$.
Since $|\{C_{k+1}, \dots, C_{t}\}| = t - k > 2^{p \cdot \numq(\phi)}$, 
\cref{prop:many-sets} implies that 
for every subset $X \subseteq V \setminus \bigcup_{k+1 \le j \le t} C_{j}$,
we have $G \models \phi(X)$ if and only if $G' \models \phi(X)$.
This implies that a sequence of vertex sets not intersecting $\bigcup_{k+1 \le j \le t} C_{j}$
is a $\TJ(\phi)$-sequence in $G$ if and only the sequence is a $\TJ(\phi)$-sequence in $G'$.
\end{proof}

The next lemma completes the proof of \cref{thm:msor2-td+k}
as it means that we have a kernel of \MSOR{1} parameterized by $\td(G)+k+|\phi|$
that preserves the minimum length of a $\TJ(\phi)$-sequence.

\begin{lemma}
\label{lem:nd-kernel}
Let $\langle \phi, G, S, S' \rangle$ be an instance of \MSOR{1} with $|S| = |S'| = k$.
In polynomial time,  one can compute a subgraph $H$ of $G$ 
such that $\dist_{\phi,G}(S,S') = \dist_{\phi,H}(S,S')$ and
the size of $H$ depends only on $\td(G) + k + |\phi|$.
\end{lemma}
\begin{proof}
Let $F$ be a treedepth decomposition of depth $\td(G)$.
If $F$ is not connected, then we add a new vertex $r$ and add edges from the new vertex to the roots of trees in $F$ 
and set $r$ to the new root. We call the resultant tree $T$.
If $F$ is connected, then we just set $T = F$ and call its root $r$.
Let $d$ be the depth of $T$. Note that $d \le \td(G) + 1$.

A node in $T$ has \emph{height} $h$ if
the maximum distance to a descendant is $h$, where the height of a leaf is $0$.
Let $c(0) = 0$, $n(0) = 1$, and for $h \ge 0$, let
\begin{align*}
  c(h+1) &= (k + 2^{n(h) \cdot \numq(\phi)}) \cdot 2^{|\phi| \cdot n(h)} \cdot 2^{(n(h) +d -h)^{2}} + 2k,
  \\
  n(h+1) &= n(h) \cdot c(h+1) + 1.
\end{align*}
In the next paragraph, we show that after exhaustively applying \cref{lem:removing-a-set} in a bottom-up manner along $T$,
each node of height $h$ has at most $c(h)$ children
and each subtree rooted at a node of height $h$ contains at most $n(h)$ nodes.
This implies that $H$ has at most $n(d)$ vertices, 
where $n(d)$ depends only on $\td(G)$, $k$, and $|\phi|$.
If $h = 0$, then the claim is trivial. Assume that the claim holds for some $h \ge 0$.
It suffices to prove the upper bound $c(h+1)$ for the number of children as the upper bound $n(h+1)$ follows immediately.
Suppose to the contrary that a node $v$ of height $h+1$ has more than $c(h+1)$ children.
Since $|S \cup S'| \le 2k$, more than $c(h+1) - 2k$ subtrees rooted at the children of $v$ have no intersection with $S \cup S'$.
Let $S_{1}, \dots, S_{p}$ be such subtrees.
By the induction hypothesis, $|V(S_{i})| \le n(h)$ holds for $i \in [p]$.
Let $R$ be the vertices on the $v$--$r$ path in $T$ (including $v$ and $r$).
Observe that, in $H$, the vertices in $V(S_{i})$ may have neighbors only in $V(S_{i}) \cup R$.
Thus the number of different types of $V(S_{1}), \dots, V(S_{p})$
is at most $2^{|\phi| \cdot n(h)} \cdot 2^{(n(h)+d-h)^{2}}$,
where $2^{|\phi| \cdot n(h)}$ is the number of possible ways for coloring $n(h)$ vertices with subsets of at most $|\phi|$ colors
and $2^{(n(h)+d-h)^{2}}$ is an upper bound on the number of different ways that
$n(h)$ vertices form a graph and have additional neighbors in $d-h$ vertices.
Since $p > c(h+1) - 2k = (k + 2^{n(h) \cdot \numq(\phi)}) \cdot 2^{|\phi| \cdot n(h)} \cdot 2^{(n(h) +d -h)^{2}}$,
there is a subset $I \subseteq [p]$ such that $|I| > k + 2^{n(h) \cdot \numq(\phi)}$ and 
all vertex sets $V(S_{i})$ with $i \in [I]$ have the same type.
This is a contradiction as \cref{lem:removing-a-set} can be applied here.

Finally, let us see how fast we can apply \cref{lem:removing-a-set} exhaustively in a bottom-up manner.
The description above immediately gives a fixed-parameter algorithm parameterized by $\td(G)+k+|\phi|$,
which is actually sufficient for our purpose.
A polynomial-time algorithm can be achieve in pretty much the same way as presented in~\cite{DemaineEHJLUU19}
for a reconfiguration problem of paths parameterized by $\td$.
The idea is to use a polynomial-time algorithm for labeled-tree isomorphism to classify subtrees into different types.
Only the differences here are that the graph is colored and the parameters involved are larger.
As the colors of the vertices can be handled by a labeling algorithm
and the involved parameters do not matter when they are too large (i.e., if it is $|V|$ or more),
we still obtain a polynomial-time algorithm.
\end{proof}

\subsection{A by-product: Cluster deletion number}

By using \cref{lem:removing-a-set} and \cref{thm:msor2-td+k},
 we can show a similar result for cluster deletion number.
For a graph $G$, a vertex subset $D$ of $G$ is a \emph{cluster deletion set}
if $G - D$ is a disjoint union of complete graphs.
The \emph{cluster deletion number} of $G$, denoted $\cd(G)$, 
is the minimum size of a cluster deletion set of $G$.
Since finding a minimum cluster deletion set is fixed-parameter tractable parameterized by 
$\cd(G)$~\cite{HuffnerKMN10}, we assume that such a set is given when $\cd(G)$ is part of the parameter.
Note that a twin cover is a cluster deletion set,
and thus $\cd(G) \le \tc(G)$ holds for every graph $G$.

\begin{corollary}
\label{cor:msor1-cd+k}
\MSOR{1} parameterized by $\cd + k + |\phi|$ is fixed-parameter tractable, where $k$ is the size of input sets.
Furthermore, for a yes instance of \MSOR{1}, finding a shortest $\TJ(\phi)$-sequence is fixed-parameter tractable with the same parameter.
\end{corollary}
\begin{proof}
Let $\langle \phi, G, S, S' \rangle$ be an instance of \MSOR{1} with $|S| = |S'| = k$
and $D$ be a cluster deletion set of $D$ with size $d \coloneqq \cd(G)$.

Let $C$ be a clique in $G - D$ such that $|C| > 2^{d+|\phi|} (k+2^{\numq(\phi)}) + 2k$,
and thus $|C \setminus (S \cup S')| > 2^{d+|\phi|} (k+2^{\numq(\phi)})$.
Since each vertex in $C$ has a subset of at most $|\phi|$ colors
and $2^{d}$ possible ways to have neighbors in $D$,
the vertices of $C \setminus (S \cup S')$,
as singletons, can be partitioned into at most $2^{d + |\phi|}$ types.
As $|C \setminus (S \cup S')| > 2^{d+|\phi|} (k+2^{\numq(\phi)})$,
at least one of the types includes more than $k+2^{\numq(\phi)}$ vertices.
For such a type, we apply \cref{lem:removing-a-set} with $p = 1$
and remove a vertex from $C \setminus (S \cup S')$.
Let $G'$ be the graph obtained by exhaustively applying this reduction.
We have $|C| \le 2^{d+|\phi|} (k+2^{\numq(\phi)}) + 2k$ for all cliques $C$ in $G'-D$.

Now observe that $G'$ has treedepth at most $d + 2^{d+|\phi|} (k+2^{\numq(\phi)}) + 2k$
since removing the $d$ vertices in $D$ cannot decrease the treedepth by more than $d$
and each clique $C$ in $G' - D$ has treedepth at most $|C|$.
Hence, \cref{thm:msor2-td+k} proves our claim.
\end{proof}


\section{PSPACE-completeness on forests of depth $3$}
\label{sec:pspace-c}

In this section, we complement \cref{thm:msor2-td+k}
by showing that if the size of input sets is not part of the parameter, then the problem becomes PSPACE-complete.

For a set $U$, a subset family $\mathcal{C} \subseteq 2^{U}$ is an \emph{exact cover} if
the elements of $\mathcal{C}$ are pairwise disjoint and $\bigcup_{C \in \mathcal{C}} C = U$.
For two exact covers $\mathcal{C}_{1}, \mathcal{C}_{2}$ of $U$,
we say that $\mathcal{C}_{1}$ can be obtained from $\mathcal{C}_{2}$ by a \emph{merge}
(and $\mathcal{C}_{2}$ can be obtained from $\mathcal{C}_{1}$ by a \emph{split})
if $\mathcal{C}_{1} \setminus \mathcal{C}_{2} = \{D_{1}\}$ and $\mathcal{C}_{2} \setminus \mathcal{C}_{1} = \{D_{2}, D_{3}\}$ 
for some $D_{1}$, $D_{2}$, $D_{3}$. 
Note that $D_{1} = D_{2} \cup D_{3}$ and $D_{2} \cap D_{3} = \emptyset$
as $\mathcal{C}_{1}$ and $\mathcal{C}_{2}$ are exact covers.

Given a set $U$, a family $\mathcal{D} \subseteq 2^{U}$,
and two exact covers $\mathcal{C}, \mathcal{C}' \subseteq \mathcal{D}$ of $U$,
\textsc{Exact Cover Reconfiguration} asks whether 
there exists a sequence $\mathcal{C}_{0}, \dots, \mathcal{C}_{\ell}$ of exact covers of $U$
from $\mathcal{C} = \mathcal{C}_{0}$ to $\mathcal{C}' = \mathcal{C}_{\ell}$ such that 
$\mathcal{C}_{i} \subseteq \mathcal{D}$ for all $i$
and $\mathcal{C}_{i}$ is obtained from $\mathcal{C}_{i-1}$ by a split or a merge for each $i \in [\ell]$.
It is known that \textsc{Exact Cover Reconfiguration} is PSPACE-complete~\cite{CardinalDEHW20}.

In this section, we prove the following hardness result
by reducing \textsc{Exact Cover Reconfiguration} to \MSOR{1}.
(Recall that \MSOR{1} belongs to PSPACE.)

\begin{theorem}
\label{thm:pspcae-c_td=3}
For some fixed $\phi$,
\MSOR{1} is PSPACE-complete on uncolored forests of depth~$3$.
\end{theorem}

\subsection{Construction}
Let $\langle U, \mathcal{D}, \mathcal{C}, \mathcal{C}' \rangle$
be an instance of \textsc{Exact Cover Reconfiguration}.
We construct an equivalent instance of \MSOR{1}.
Without loss of generality, we assume that $U$ is a set of positive integers greater than or equal to $3$.\footnote{%
We want to identify elements of $U$ with positive integers but do not want to use $1$ and $2$ for some technical reasons,
which will be clear in the proof.}

For each set $D \in \mathcal{D}$, we construct a tree $T_{D}$ as follows (see \cref{fig:treeD}).
The tree $T_{D}$ contains a central vertex called the \emph{root}.
For each $d \in D$, the root has a child with $d$ grandchildren.
We call the subtree rooted at a child of the root a \emph{star} and each leaf in a star a \emph{star leaf}.
Additionally, the root has two more children that have degree~$1$. They are called the \emph{antennae}.

\begin{figure}[htb]
  \centering
  \includegraphics[scale=1]{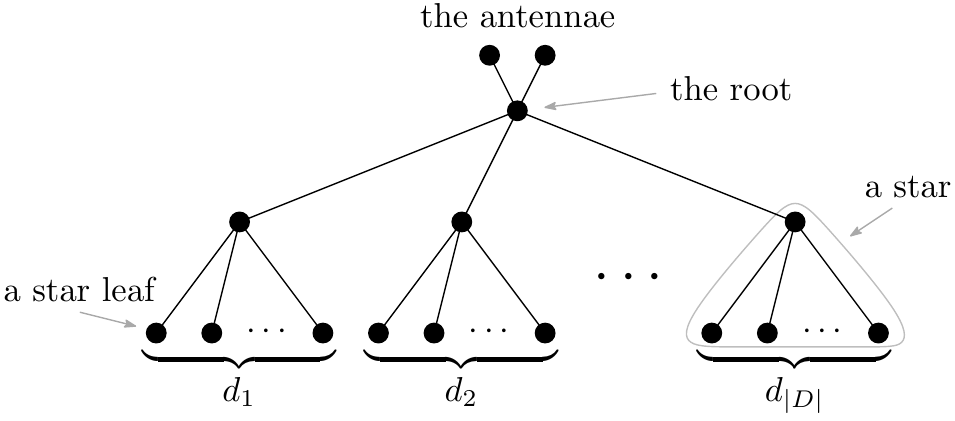}
  \caption{The tree $T_{D}$ with $D = \{d_{1}, d_{2}, \dots, d_{|D|}\}$.}
  \label{fig:treeD}
\end{figure}

The entire forest $F$ consists of trees $T_{D}$ for all $D \in \mathcal{D}$
and eight isolated vertices. By $I$, we denote the set of the isolated vertices.
Clearly, $F$ has treedepth~$3$.
The initial set $S$ consists of all vertices in $I$
and all star leaves of $T_{D}$ for all $D \in \mathcal{C}$.
Similarly, 
the target set $S'$ consists of all vertices in $I$
and all star leaves of $T_{D}$ for all $D \in \mathcal{C}'$.
Note that $|S| = |S'| = |U| + 8$.

For a set $R \subseteq V(F)$ and a set $D \in \mathcal{D}$,
the tree $T_{D}$ is \emph{full} (\emph{empty}) \emph{under} $R$
if $R$ contains all (no, resp.) star leaves of $T_{D}$,
and $T_{D}$ is \emph{clean} if it is full or empty.
We also say that, for a set $R \subseteq V(F)$ and a set $D \in \mathcal{D}$,
a star in $T_{D}$ is \emph{full} (\emph{empty})
if $R$ contains all (no, resp.) star leaves of the star,
and the star is \emph{clean} if it is full or empty.

A tree $T_{D}$ is \emph{marked} by a vertex set if both antennae are included in the vertex set.
A star in $T_{D}$ is \emph{marked} by a vertex set if the center (i.e., the unique non-leaf vertex) of the star is included.
We use the eight additional vertices to mark three trees and two stars.

We construct an \mso{1} formula $\phi(X)$ expressing that $X$ satisfies one of the following two conditions.
\begin{enumerate}
  \item Exactly eight vertices in $X$ are not star leaves and 
  all trees $T_{D}$ are clean.

  \item Exactly eight vertices in $X$ are not star leaves,
  exactly three trees $T_{D_{1}}$, $T_{D_{2}}$, $T_{D_{3}}$ are marked, all other trees are clean, 
  and the following conditions are satisfied. 
  \begin{itemize}
    \item One of the trees, say $T_{D_{1}}$, is not clean and another one, say $T_{D_{3}}$, is clean.
    \item Exactly two stars are marked, one in $T_{D_{1}}$ and the other in $T_{D_{2}}$.
    The marked star in $T_{D_{1}}$ is clean if and only if so is the marked star in $T_{D_{2}}$.
    \item All unmarked stars are clean.
  \end{itemize}
\end{enumerate}
Constructing such $\phi(X)$ is tedious but not difficult.
The expression is given in \cref{ssec:phi-for-hardness}.

We say that a vertex set $X$ is \emph{clean} if the first condition above is satisfied.
Since exactly eight vertices in $X$ are not star leaves in both case,
every $\TJ(\phi)$-move involves either two star leaves or two vertices that are not star leaves.
The definition of $\phi$ also implies that a token on an antenna can move only if the set is clean.
Also, a token on a star center can move only if the star is clean.

In what follows, we show that the constructed instance $\langle \phi, F, S, S' \rangle$ is a yes-instance of \MSOR{1}
if and only if
$\langle U, \mathcal{D}, \mathcal{C}, \mathcal{C}' \rangle$ is a yes-instance of \textsc{Exact Cover Reconfiguration}.

\subsection{The if direction}
Assume that $\langle U, \mathcal{D}, \mathcal{C}, \mathcal{C}' \rangle$ is a yes-instance of \textsc{Exact Cover Reconfiguration}.
Let $\mathcal{C}_{0}, \dots, \mathcal{C}_{\ell}$ be a reconfiguration sequence
from $\mathcal{C} = \mathcal{C}_{0}$ to $\mathcal{C}' = \mathcal{C}_{\ell}$ by splits and merges.
For $0 \le i \le \ell$, let $S_{i}$ be the set consists of
the eight vertices in $I$ and all star leaves in $T_{D}$ for all $D \in \mathcal{C}_{i}$.
Observe that $S_{0} = S$, $S_{\ell} = S'$, and each $S_{i}$ is clean.
Thus, it suffices to show that 
there is a $\TJ(\phi)$-sequence from $S_{i-1}$ to $S_{i}$ for all $i \in [\ell]$.

Let $i \in [\ell]$. Since $\TJ(\phi)$-sequences are reversible,
it suffices to consider the case where $\mathcal{C}_{i}$ is obtained from $\mathcal{C}_{i-1}$ by a split.
Let $\mathcal{C}_{i-1} \setminus \mathcal{C}_{i} = \{D_{1}\}$ and $\mathcal{C}_{i} \setminus \mathcal{C}_{i-1} = \{D_{2}, D_{3}\}$.
A $\TJ(\phi)$-sequence from $S_{i-1}$ to $S_{i}$ can be constructed as follows.
(Recall that $D_{1} = D_{2} \cup D_{3}$ and $D_{2} \cap D_{3} = \emptyset$.)
\begin{enumerate}
  \setcounter{enumi}{-1}
  \item Choose arbitrary six tokens on $I$ and call them the \emph{antenna tokens}, and the remaining two the \emph{star tokens}.
  \item Move the antenna tokens to the antennae of $T_{D_{1}}$, $T_{D_{2}}$, $T_{D_{3}}$.
  \item For each $d \in D_{2}$ and $j \in \{1,2\}$, let $S_{d}^{(j)}$ be the star in $T_{D_{j}}$ corresponding to~$d$.
  \label{itm:D2}
  \begin{enumerate}
    \item Move the star tokens to the centers of $S_{d}^{(1)}$ and $S_{d}^{(2)}$.
    \item Move the tokens on the star leaves of $S_{d}^{(1)}$ to the star leaves of $S_{d}^{(2)}$.
  \end{enumerate}

  \item Replace $D_{2}$ with $D_{3}$ and execute Step~\ref{itm:D2}.
  
  \item Move the antenna tokens and the star tokens to $I$.
\end{enumerate}
We can see that $S_{i}$ is obtained after the modifications above, and all intermediate sets satisfy~$\phi$.
Hence, this direction holds.

\subsection{The only-if direction}
Assume that $\langle \phi, F, S, S' \rangle$ is a yes-instance of \MSOR{1}.
Let $S_{0},\dots,S_{\ell}$ be a shortest $\TJ(\phi)$-sequence from $S = S_{0}$ to $S' = S_{\ell}$.
For each $S_{i}$, let $\mathcal{S}_{i} = \{D \in \mathcal{D} \mid T_{D} \ \text{is full under} \ S_{i}\}$.
Observe that $\mathcal{S}_{0} = \mathcal{C}$ and $\mathcal{S}_{\ell} = \mathcal{C}'$.
Since $d > 1$ for each $d \in U$,
a single $\TJ(\phi)$-move cannot make a full star empty.
This implies the following.
\begin{observation}
\label{obs:both-clean}
If $S_{i}$ and $S_{i+1}$ are clean, then $\mathcal{S}_{i} = \mathcal{S}_{i+1}$.
\end{observation}

The next lemma is the main technical ingredient in this direction.
\begin{lemma}
\label{lem:nonclean-interval}
If $S_{i}$ and $S_{i'}$ are clean, $i + 1 < i'$, and $S_{j}$ is not clean for all $j$ with $i < j < i'$,
then $\mathcal{S}_{i'}$ can be obtained from $\mathcal{S}_{i}$ by a split or a merge.
\end{lemma}
\begin{proof}
Since $S_{i}$ is clean and $S_{i+1}$ is not,
the $\TJ(\phi)$-move from $S_{i}$ to $S_{i+1}$ involves two star leaves.
In $S_{i}$, the following conditions are satisfied:
\begin{itemize}
  \item all trees $T_{D}$ are clean;
  \item exactly three trees $T_{D_{1}}$, $T_{D_{2}}$, $T_{D_{3}}$ are marked;
  \item exactly two of them, say $T_{D_{1}}$ and $T_{D_{2}}$, contain one marked star each.
\end{itemize}
The second and third conditions above hold
as the next set $S_{i+1}$ satisfies $\phi$ but is not clean.

For $j \in \{1,2\}$, let $R_{j}$ be the marked star in $T_{D_{j}}$ under $S_{i}$.
We can see that $S_{i+1}$ is obtained from $S_{i}$
by moving a token in one of the stars $R_{1}$ and $R_{2}$ to the other.
By symmetry, assume that a token is moved from $R_{1}$ to $R_{2}$.
This implies that $T_{D_{1}}$ is full and $T_{D_{2}}$ is empty under $S_{i}$
and in particular that $R_{1}$ is full and $R_{2}$ is empty under~$S_{i}$.

Let us consider what kind of moves the sequence may take in the next steps.
Since $R_{1}$ and $R_{2}$ are not clean under $S_{i+1}$,
the definition of $\phi$ allows us to move only the tokens in the star leaves of $R_{1}$ and $R_{2}$ until both of them become clean.
Furthermore, the last condition of $\phi$ asks to make them clean at the same time.
This implies that $R_{1}$ and $R_{2}$ have to correspond to the same integer $d \in U$.
If we make $R_{1}$ full and $R_{2}$ empty, then we obtain $S_{i}$ again.
This contradicts the assumption that the sequence is shortest.
Therefore, we can conclude that the next $d-1$ steps move the tokens on star leaves in $R_{1}$ to $R_{2}$,
and make $R_{1}$ empty and $R_{2}$ full.
Now we can move tokens at the centers of $R_{1}$ and $R_{2}$ to star centers in the same trees.

Until one of $T_{D_{1}}$ and $T_{D_{2}}$ becomes clean,
we have to repeat the same steps of moving tokens
from the star leaves in a star to the star leaves in another star corresponding to the same integer.
If we make $T_{D_{1}}$ full (and thus $T_{D_{2}}$ empty), then we obtain $S_{i}$ again, contradicting the assumption.
Thus, we make $T_{D_{1}}$ empty or $T_{D_{2}}$ full.
Since $D_{1} \ne D_{2}$, it is impossible to make $T_{D_{1}}$ empty and $T_{D_{2}}$ full at the same time.
Thus we have the following two cases to consider:
\begin{enumerate}
  \item $T_{D_{1}}$ is empty and $T_{D_{2}}$ is not full;
  \item $T_{D_{1}}$ is not empty and $T_{D_{2}}$ is full.
\end{enumerate}
We claim that the first case corresponds to a merge and the second to a split.
In the following, we only consider the second case (corresponding to a split)
as the other case is symmetric and allows almost the same proof.

In the current situation,
$T_{D_{1}}$ itself is not clean but its stars are clean,
$T_{D_{2}}$ is full, and $T_{D_{3}}$ is clean.
Recall that the sequence eventually reaches the clean set $S_{i'}$ and
that unless all trees are clean,
no tokens on the antennae of $T_{D_{1}}$, $T_{D_{2}}$, $T_{D_{3}}$ can move.
Thus, in the next steps, we have to make all $T_{D_{1}}$, $T_{D_{2}}$, $T_{D_{3}}$ clean.
If $T_{D_{3}}$ is full now, then this can be possible only by moving back the tokens in $T_{D_{2}}$ to $T_{D_{1}}$ (as $D_{3} \ne D_{2}$).
This contradicts the assumption that the sequence is shortest, and thus we conclude that $T_{D_{3}}$ is empty at this moment.
Similarly, we cannot make $T_{D_{2}}$ empty as it is only possible by making $T_{D_{1}}$ full and leaving $T_{D_{3}}$ empty.
Hence, the only option is to make $T_{D_{1}}$ empty and $T_{D_{3}}$ full, while leaving $T_{D_{2}}$ full.
We can observe that this is possible only if $D_{1} \setminus D_{2} = D_{3}$
by applying almost the same argument for the first steps for making $T_{D_{2}}$ full.
Therefore, it holds that $\mathcal{S}_{i'} = (\mathcal{S}_{i} \setminus \{D_{1}\}) \cup \{D_{2}, D_{3}\}$,
i.e., $\mathcal{S}_{i'}$ can be obtained from $\mathcal{S}_{i}$ by a split.
\end{proof}

Let $S_{i_{0}}, \dots, S_{i_{p}}$
be the sequence obtained from $S_{0}, \dots, S_{\ell}$ by skipping non-clean sets.
By \cref{obs:both-clean,lem:nonclean-interval},
in the corresponding sequence $\mathcal{S}_{i_{0}}, \dots, \mathcal{S}_{i_{p}}$,
consecutive families $\mathcal{S}_{i_{j-1}}$ and $\mathcal{S}_{i_{j}}$ are either the same or in the split-merge relation
for every $j \in [p]$. Furthermore, $\mathcal{S}_{i_{0}} = \mathcal{S}_{0} = \mathcal{C}$
and $\mathcal{S}_{i_{p}} = \mathcal{S}_{\ell} = \mathcal{C}'$.
Therefore, there is a reconfiguration sequence from $\mathcal{C}$ to $\mathcal{C}'$ by splits and merges.
This completes the proof of \cref{thm:pspcae-c_td=3}.

\subsection{Constructing $\phi(X)$}
\label{ssec:phi-for-hardness}

Recall that the \mso{1} formula $\phi(X)$ expresses the property that $X \subseteq V(F)$ satisfies
one of the following two conditions.
\begin{enumerate}
  \item Exactly eight vertices in $X$ are not star leaves and 
  all trees $T_{D}$ are clean.

  \item Exactly eight vertices in $X$ are not star leaves,
  exactly three trees $T_{D_{1}}$, $T_{D_{2}}$, $T_{D_{3}}$ are marked, all other trees are clean, 
  and the following conditions are satisfied.
  \begin{itemize}
    \item One of the trees, say $T_{D_{1}}$, is not clean and another one, say $T_{D_{3}}$, is clean.
    \item Exactly two stars are marked, one in $T_{D_{1}}$ and the other in $T_{D_{2}}$.
    The marked star in $T_{D_{1}}$ is clean if and only if so is the marked star in $T_{D_{2}}$.
    \item All unmarked stars are clean.
  \end{itemize}
\end{enumerate}

We call a set satisfying the first condition \emph{clean}
and a set satisfying the second condition \emph{almost clean}.
In the following, we define formulas $\textsf{clean}(X)$ and $\textsf{almost-clean}(X)$
expressing clean sets and almost clean sets, respectively, and set
\[
  \phi(X) = \textsf{clean}(X) \lor \textsf{almost-clean}(X).
\]

\subsubsection{Auxiliary subformulas}

We first define $u\textsf{-leaf}(v)$ meaning that $v$ is a leaf (a degree-$1$ vertex) attached to $u$
and $\textsf{dist2}(r,s)$ meaning that the distance between $r$ and $s$ is exactly $2$:
\begin{align*}
  u\textsf{-leaf}(v) &= E(u,v) \land (\forall w \colon E(w,v) \Rightarrow w = u),
  \\
  \textsf{dist2}(r,s) &= (r \ne s) \land \lnot E(r,s) \land (\exists u \colon E(r,u) \land E(u,s)).
\end{align*}
Note that $u$ in $u\textsf{-leaf}(v)$ is a free variable too.

Now the property of being the root, center, or a star leaf in a tree $T_{D}$ can be expressed as follows:
\begin{align*}
  \textsf{root}(r)
  &= 
  \dot{\exists} u,v \colon
  r\textsf{-leaf}(u) \land
  r\textsf{-leaf}(v) \land
  (\forall w \colon r\textsf{-leaf}(w) \Rightarrow ((w = u) \lor (w = v))),
  \\
  \textsf{center}(c) &= \exists r \colon \textsf{root}(r) \land E(r,c) \land \lnot r\textsf{-leaf}(c),
  \\
  \textsf{star-leaf}(v) 
  &= 
  \exists r \colon \textsf{root}(r) \land \textsf{dist2}(r,v).
\end{align*}
Here we use the characterizations that
a root has exactly two leaves attached (recall that each star contains at least three star leaves),
a center is a non-leaf vertex adjacent to a root, and 
a star leaf has distance~$2$ to a root.

Assuming that $r$ is a root vertex,
the next formulas express that the tree $T_{D}$ rooted at $r$ is full, empty, or clean in $X$, respectively:
\begin{align*}
  \textsf{full-tree}(r,X)  &= \forall s \colon \textsf{dist2}(r,s) \Rightarrow s \in X, \\ 
  \textsf{empty-tree}(r,X) &= \forall s \colon \textsf{dist2}(r,s) \Rightarrow \lnot (s \in X),\\
  \textsf{clean-tree}(r,X) &= \textsf{full-tree}(r,X) \lor \textsf{empty-tree}(r,X).
\end{align*}

Similarly, assuming that $c$ is a center of some star,
the next formulas express that the star with the center $c$ is full, empty, or clean in $X$, respectively:
\begin{align*}
  \textsf{full-star}(c, X)  &= \forall s \colon c\textsf{-leaf}(s) \Rightarrow s \in X, \\ 
  \textsf{empty-star}(c, X) &= \forall s \colon c\textsf{-leaf}(s) \Rightarrow \lnot (s \in X), \\
  \textsf{clean-star}(c, X) &= \textsf{full-star}(c, X) \lor \textsf{empty-star}(c, X).
\end{align*}

Assuming that $r$ is the root vertex of a tree $T_{D}$,
the next formula means that $T_{D}$ is marked as the two antennae are only leaves attached to $r$:
\[
  \textsf{marked}(r,X) = \dot{\exists} u, v \in X \colon r\textsf{-leaf}(u) \land r\textsf{-leaf}(v).
\]

Finally, the following formula means that $X$ contains exactly eight vertices that are not star leaves.
\begin{align*}
  &\textsf{eight-non-star-leaves}(X)
  = \dot{\exists} v_{1}, \dots, v_{8} \in X \colon  \\
  &\qquad
  \textstyle\bigwedge_{1 \le i \le 8} \lnot \textsf{star-leaf}(v_{i})
  \land (\forall v \colon \lnot \textsf{star-leaf}(v) \Rightarrow \textstyle\bigvee_{1 \le i \le 8} v = v_{i}).
\end{align*}

\subsubsection{Main subformulas}
Given subformulas above, expressing $\textsf{clean}(X)$ is straightforward:
\begin{align*}
  \textsf{clean}(X) 
  =
  \textsf{eight-non-star-leaves}(X)
  \land
  \left(\forall r \colon \textsf{root}(r) \Rightarrow \textsf{clean-tree}(r, X)\right).
\end{align*}

The expression of $\textsf{almost-clean}(X)$ is a little more involved
but seeing the equivalence to the conditions is not difficult:
\begin{align*}
  &\textsf{almost-clean}(X) =
  \textsf{eight-non-star-leaves}(X) \land {} \\
  &\quad
  (\dot{\exists} r_{1}, r_{2}, r_{3} \colon
  \textstyle\bigwedge_{1 \le i \le 3} (\textsf{root}(r_{i}) \land \textsf{marked}(r_{i},X))
  \\
  &\quad \land
  \lnot \textsf{clean-tree}(r_{1}, X) \land
  \left(\forall r \colon \textsf{root}(r) \land \lnot \textsf{clean-tree}(r, X) \Rightarrow (r = r_{1}) \lor (r = r_{2})\right)
  \\
  &\quad\land
  (\dot{\exists} c_{1}, c_{2} \in X \colon 
  \textsf{center}(c_{1}) \land \textsf{center}(c_{2}) 
  \land E(c_{1}, r_{1}) \land E(c_{2}, r_{2}) \\
  & \qquad\qquad \land (\forall c \colon (\textsf{center}(c) \land (c \ne c_{1}) \land (c \ne c_{2})) \Rightarrow \textsf{clean-star}(c, X)) \\
  & \qquad\qquad \land (\textsf{clean-star}(c_{1}, X) \Leftrightarrow \textsf{clean-star}(c_{2}, X)))).
\end{align*}

\section{Remarks on the token sliding setting}
\label{sec:token-sliding}

Here we consider the token sliding variants of \MSOR{1} and \MSOR{2}
that require the exchanged vertices to be adjacent in each reconfiguration step.

\cref{thm:msor1-nd} says \MSOR{1} is fixed-parameter tractable parameterized by $\nd + |\phi|$.
Unfortunately, this result strongly depends on the token jumping setting.
Especially, \cref{lem:nd-sameshape} fails to hold for the token sliding setting.
Minor modifications do not look quite promising.

\cref{thm:msor2-td+k}, which says that \MSOR{2} is fixed-parameter tractable parameterized by $\td + k + |\phi|$,
can be easily modified for the token sliding setting. Actually, the proof works almost as it is.
Although the proof of \cref{lem:only-k-sets} depends on the token jumping setting,
it can be easily adopted to the token sliding setting.
Actually, if we start with a token sliding sequence, we end up with a token sliding sequence after the same modification in the proof.
No other part depends on the reconfiguration rule.
An immediate corollary to this observation is that 
\textsc{Independent Set Reconfiguration} under the token sliding rule is fixed-parameter tractable parameterized by $\td + k$.
This particular result was shown independently by Bartier et al.~\cite{BartierBM22}

The hardness shown in \cref{thm:pspcae-c_td=3} can be easily modified for the token sliding setting if we slightly weaken it.
To the forest $F$ constructed there, we add a universal vertex $u$ adjacent to all other vertices.
Then we define the feasibility of a set $X$ as 
either (i) $X$ satisfies the original formula $\phi$ in $F$, or
(ii) $X$ is obtained from a set $X'$ satisfying $\phi$ in $F$ by exchanging a member of $X'$ with the universal vertex $u$.
Constructing an \mso{1} formula expressing this condition is easy and proving the equivalence is straightforward.
This implies that the token sliding variant of 
\MSOR{1} is PSPACE-complete even on graphs of treedepth~$4$.


\section{Conclusion}
\label{sec:conclusion}

In this paper, we revisited the reconfiguration problems of vertex sets defined by \mso{} formulas,
while putting the length constraint of reconfiguration sequence aside.
We showed that the problem is fixed-parameter tractable parameterized solely by neighborhood diversity
and by the combination of treedepth and the vertex-set size.
The parameterization solely by treedepth would not work
since the problem is PSPACE-complete on forests of depth~$3$ as we showed.

Given the positive result for neighborhood diversity and the known hardness for clique-width
(implied by the one for bandwidth~\cite{Wrochna18}),
a natural target would be an extension to modular-width,
which is a parameter sitting between neighborhood diversity and clique-width (see \cref{fig:parameters}).
It is known that a special case, the independent set reconfiguration,
is fixed-parameter tractable parameterized by modular-width~\cite{BelmonteHLOO20},
but the algorithm in \cite{BelmonteHLOO20} is already quite nontrivial.

Another direction would be strengthening the hardness for treedepth.
In \cref{sec:pspace-c}, we showed the hardness for a quite complicated and rather unnatural formula~$\phi$, 
which simulates the merge and split operations.
Although this rules out the possibility of meta-theorems parameterized by treedepth,
it would be still interesting to investigate the complexity of specific more natural problems.
For example, what is the complexity of the independent set reconfiguration
and the dominating set reconfiguration parameterized solely by treedepth?



\bibliographystyle{siamplain}
\bibliography{mso-reconf}

\end{document}